%% file: majo_measure.tex
\documentclass[
12pt]{iopart}
\usepackage{harvard}
\expandafter\let\csname equation*\endcsname\relax
\expandafter\let\csname endequation*\endcsname\relax
\usepackage[fleqn]{amsmath}
\usepackage{graphicx}% Include figure files
\usepackage{dcolumn}% Align table columns on decimal point
\usepackage{bm}% bold math
\usepackage{bbm}
\usepackage{amsthm}
\usepackage{dsfont}
\usepackage{amssymb}
\usepackage{array}
\usepackage{wrapfig}
\usepackage[latin10]{inputenc}
\usepackage{amsfonts}
\usepackage[english]{babel}

\renewcommand{\harvardurl}[1]{\textbf{URL:} \url{#1}}
\providecommand\phantomsection{}

\newcommand{\bra}[1]{\langle #1|}
\newcommand{\ket}[1]{|#1\rangle}

\newcommand{\id}{\mathbbm{1}}

\newtheorem{thm}{Theorem}

\newtheorem{lem}[thm]{Lemma}
\newtheorem{defin}{Definition}
\newtheorem*{coro}{Corollary}
\newtheorem*{thmnonumb}{Theorem}
\DeclareMathOperator{\md}{d}
\DeclareMathOperator{\rmsupp}{supp}

\newcommand{\smoothhmax}{H_{\max}^\varepsilon}

\let\oldmarginpar\marginpar
\renewcommand\marginpar[1]{\-\oldmarginpar[\raggedleft\marginparsize #1]%
{\raggedright\marginparsize #1}}
\begin{document}

\setlength{\tabcolsep}{1ex}%for latex tables

\title{A measure of majorisation emerging from single-shot statistical mechanics}% Force line breaks with \\
\author{D Egloff$^1$ $^4$\footnote[1]{These authors contributed equally to this work}, O C O Dahlsten$^2$ $^3$\footnotemark[1] $^1$\footnote{Corresponding author: oscar.dahlsten@physics.ox.ac.uk
}, R Renner$^1$ and V Vedral$^2$ $^3$}
\address{$^1$ Institute for Theoretical Physics, ETH Z\"urich, 8093 Zurich, Switzerland}

\address{$^2$ Atomic and Laser Physics, Clarendon Laboratory,
University of Oxford, Parks Road, Oxford OX13PU, United Kingdom}
\address{$ ^3$ Center for Quantum Technologies, National University of Singapore, Republic of Singapore}
\address{$^4$  Institute for Theoretical Physics, Universit\"at Ulm,  89069 Ulm, Germany}
\date{\today}

\begin{abstract}
The use of the von Neumann entropy in formulating the laws of thermodynamics has recently been challenged. It is associated with the average work whereas the work guaranteed to be extracted in any single run of an experiment is the more interesting quantity in general. We show that an expression that quantifies majorisation determines the optimal guaranteed work. We argue it should therefore be the central quantity of statistical mechanics, rather than the von Neumann entropy. In the limit of many identical and independent subsystems (asymptotic i.i.d) the von Neumann entropy expressions are recovered but in the non-equilbrium regime the optimal guaranteed work can be radically different to the optimal average. Moreover our measure of majorisation governs which evolutions can be realized via thermal interactions, whereas the nondecrease of the von Neumann entropy is not sufficiently restrictive. Our results are inspired by single-shot information theory.
\end{abstract}

\maketitle

\phantomsection{}\addcontentsline{toc}{subsection}{Introduction}
{\large S}tatistical mechanics is a corner-stone of modern physics. Many of its basic paradigms and mathematical methods were set in an era where the experimental abilities were much more limited and modern information theory not developed. Accordingly there is currently significant momentum in investigating the theory's foundations in the quantum and nano regimes, see e.g.~\cite{Jarzynski97,Lloyd97,GemmerM2004,AllahverdayanBN04,LindenPS09,ToyabeSUMS10,BrandaoHORS11,JenningsRHNM12} to mention but a few recent contributions. We here derive an alternative type of statistical mechanics from scratch. Our approach is inspired by recent results in information theory~\cite{Renner05,RennerW04} and builds on~\cite{DahlstenRRV11,delRioARDV11,Aberg11,HorodeckiO11}. We argue this approach is both significantly more general than the standard theory and addresses questions more relevant to modern experiments. 

It is more general in that we will not assume that the states of systems of interest are thermal, but rather just that there is a heat bath which when interacting with a system gradually takes that system towards a thermal state. Thus the system of interest is not necessarily in equilibrium. In fact we will allow for any probability distribution over energy levels. We do in particular not assume that the system under consideration is large or that internal correlations are negligible. This makes the approach significantly more relevant to modern experiments where small sub-systems can be addressed individually and in time-scales faster than the thermalisation time.

A key difference regarding which questions are addressed is that we focus not on averages of distributions as in standard statistical mechanics. Instead we ask,  for any given single run of an experiment, which threshold values are guaranteed to be exceeded, or more generally guaranteed to be exceeded up to some probability $\varepsilon$, not necessarily small. This is referred to as the {\em single-shot} paradigm, as opposed to the average paradigm. This distinction is important when distributions of quantities have a significant spread around the average, as is often the case for small systems. 

To see why we choose the single-shot paradigm, consider work extraction from a system. Work is a particularly important quantity, appearing in the first and second laws of thermodynamics and of crucial importance in the context of engines. As usually this is the case, let there be more than one way to extract work, e.g. different ways of changing the Hamiltonian of the system from which work is to be extracted. Say for concreteness that there are two different strategies: strategy 1 (S1) and strategy 2 (S2). Let S1 (S2) be associated with probability distributions over extracted work $w$ denoted by $p_1(W)\, (p_2(W))$. Suppose that the averages are equal, i.e. $\langle W\rangle_{S1}=\langle W\rangle_{S2}$, but $p_1(W)$ has no spread around the average, whereas $p_2(W)$ has a significant spread. Are these protocols now equally `good', as one might think by looking at the averages? This is certainly not the case in general. Suppose that there is a threshold for $W$, $W^*$ that needs to be exceeded. Such thresholds often exist as e.g. an activation energy for some process, or a band-gap to jump. Suppose moreover, to make this example interesting, that $\langle W\rangle_{S1}=\langle W\rangle_{S2}>W^*$. Now with S1 we will indeed achieve the threshold with probability 1, but with S2 the probability of exceeding the threshold can be arbitrarily small, as there may be a small probability of significantly exceeding the threshold but a large probability of just about failing to achieve it(!) 

If we instead of the average considered the work guaranteed up to probability $\varepsilon$, writing this as 
$W^\varepsilon_S$, where $S$ is the strategy, we see that $W^\varepsilon_{S1}=\langle W\rangle_{S1}>W^*\,\forall \varepsilon\in[0,1]$ whereas $W^\varepsilon_{S2}<W^*$ for all $\varepsilon$ smaller than whatever the probability of being below the threshold is. This example demonstrates that the single-shot quantity $W^\varepsilon_S$ does, in contrast to the average $\langle W\rangle_{S}$, make it clear that the two protocols perform very differently. We find this example most interesting if one considers different $\varepsilon$ and not only $\varepsilon=0$.

In this article we derive an expression concerning the optimal work $W^\varepsilon_S$ for various initial and final conditions. More specifically we consider a system with an initial Hamiltonian $H_i$ and density matrix $\rho$, and a given final Hamiltonian $H_f$ and density matrix $\sigma$. We only consider states $\rho$ and $\sigma$ diagonal in the energy basis. The experimenter may choose from a set of possible strategies $S$, which are arbitrary combinations of infinitessimal changes in the Hamiltonian, and interactions with a thermalising heat bath associated with temperature $T$. The work guaranteed to be exceeded with a failure probability up to $\varepsilon$ is then written as $W^\varepsilon_S(\rho,H_i\rightarrow \sigma, H_f)$.  As the main technical result of this paper we derive an expression for the optimal guaranteed work: $W^\varepsilon(\rho,H_i\rightarrow \sigma, H_f)=\max_S W^\varepsilon_S(\rho,H_i\rightarrow \sigma, H_f)$. We show it is given---if we suppress certain details to be specified later---by 
\begin{eqnarray*}
	W^\varepsilon(\rho,H_i\rightarrow \sigma, H_f)=kT\ln{{\bf\sf M(G(\rho,H_i)||G(\sigma,H_f))}},
\end{eqnarray*} 
where ${\bf\sf M(G(\rho,H_i)||G(\sigma,H_f))}$ is a measure of how much $\rho$ {\em majorises} $\sigma$. This measure of majorisation emerges from our considerations. A way of calculating the deterministic work for the zero-risk case in terms of diagrams has been given in~\cite{HorodeckiO11}. In this case the results coincide. In \cite{Aberg11} deterministic work is defined as work that will be extracted, no more no less, with probability $1$. $(\epsilon,\delta)$-deterministic work $W$ means the work will be in the interval  $[W-\delta,W+\delta]$ up to an error probability of $\epsilon$. Here in contrast we have considered guaranteed work. The difference between guaranteed and deterministic work can be most easily seen for $\epsilon$ and $\delta$ both being $0$. Then having non-zero deterministic work necessitates no spread in the distribution whereas guaranteed work means that the spread lies above the wanted threshold. One can get an upper bound for the deterministic work by the guaranteed work, but in general they are different objects.

In standard thermodynamics it is the free energy difference $\Delta F=\Delta (U-TS_{\mathrm{vN}})$ which determines the optimally extractable work, and moreover gives a criterion for which state transformations are realizable by interactions with a heat bath, via $\Delta F\leq 0$, as can be shown to be true for many reasonable models of thermalisation.  We argue however that ${\bf \sf M}$ should be the central quantity of statistical mechanics, by virtue of: (i) characterising optimal guaranteed work and (ii) providing a tight condition for which evolutions are consistent with our thermalisation model, as opposed to $\Delta F\leq 0$ which we show is necessary but not sufficient. These statements will be made precise later in this Letter. We call ${\bf \sf M}$ the {\em relative mixedness}. In certain limits ${\bf \sf M}$ reduces to differences in entropy of so-called single-shot entropies, which in turn in the asymptotic i.i.d. limit ($\rho^{\otimes n}$, $n\rightarrow \infty$) reduce to the von Neumann entropy $S_{\mathrm{vN}}$. But in general the relative mixedness of two states can be very different to the standard free energy difference $\Delta F$.   

We go on to make use of the results relating to the relative mixedness to formulate the laws of thermodynamics in the single-shot paradigm.  The first law is modified to be about guaranteed work rather than average work. Several versions of the second law are all modified in important ways. Apart from the already mentioned replacement of free energy decrease, the optimal extractable work turns out not to be a function of state but a relative notion between two states. The relative mixedness acts as a unifying feature which means that the new laws nevertheless have a simple structure.

As there are strong connections between the structure of entanglement theory and that of thermodynamics, we moreover consider the impact on entanglement theory, showing how to quantify entanglement as a relative notion between two states using relative mixedness rather than as a state function given by the von Neumann entropy.

\section*{RESULTS}\phantomsection{}\addcontentsline{toc}{section}{Results}
\phantomsection{}\addcontentsline{toc}{subsection}{Existing results}
{\bf Existing results.} We begin with briefly reviewing key results that we shall later recover as special cases of our expression. (This is thus not an exhaustive list of all previous results).  The results concern extracting work in the presence of a heat bath at temperature $T$ . The details of the models of work extraction in the different papers are not a priori identical, but we shall recover the same expressions within the model here.

In~\cite{DahlstenRRV11} an $n$-cylinder Szilard engine was considered and the following expression derived:
\begin{eqnarray}
\label{eq:firstpaper}
W^\varepsilon=\left(n-\smoothhmax \right)kT\ln2.
\end{eqnarray}
 Here $W^\varepsilon$ is the work that can be extracted in a process with maximum probability of failure $\varepsilon$. $\smoothhmax$ is the {\em smooth max entropy} of the density matrix representing a work-extracting agent's initial knowledge about the state of the working medium. This is defined as $\smoothhmax(\rho )=\log \left(\text{rank}^\varepsilon(\rho)\right)$, with $\text{rank}^\varepsilon(\rho)$ the number of non-zero eigenvalues minimised over all states within $\varepsilon$ trace distance of $\rho$. (Actually there is an alternative definition as well but they are both known to coincide up to an additive $\log\frac{1}{\varepsilon}$ term, so for simplicity  we focus on one definition here.) $T$ is as mentioned above the temperature of the heat bath, and $k$ Boltzmann's constant.  $\smoothhmax (\rho)$ reduces to the von Neumann entropy in the  in the i.i.d. limit, i.e., when $\rho=\tau^{\otimes n}$, $n\rightarrow \infty$ and $\varepsilon \rightarrow 0$. Physically this corresponds to systems composed of very large numbers of identical and uncorrelated subsystems.

A key result obtained independently in the more recent papers~\cite{Aberg11,HorodeckiO11} is that given an initial state $\rho$ and a final thermal state $\rho_T$ over the same energy levels, the work that can be extracted given access to a heat bath of temperature $T$, and with up to $\varepsilon$ failure probability is:
\begin{equation}
\label{eq:AbergOppenheim}
W^\varepsilon =kT\ln (2) D_0^\varepsilon (\rho||\rho_T), 
\end{equation}
where $D_0^\varepsilon(\rho||\rho_T)$ is the $\varepsilon$-smooth relative entropy of order 0 (see~\cite{Datta09}). In~\cite{Aberg11} $\rho$ is taken to be diagonal in the energy eigenbasis and in the a priori distinct set-up in \cite{HorodeckiO11} the state if not already diagonal in the energy eigenbasis may be replaced by the corresponding diagonal (decohered) state without changing the expression for the extractable work (in \cite{HorodeckiO11} also the probabilistic work for the opposite process was given and the deterministic work for arbitrary (initially energy-diagonal) state conversion). The RHS of Eq.~\ref{eq:AbergOppenheim} reduces to $W=kT\ln (2) D(\rho||\rho_T)$ for the standard relative entropy in the  asymptotic i.i.d.\ (von Neumann entropy) regime. That latter expression is well-established, see e.g.~\cite{Donald87}. Eq.~\ref{eq:AbergOppenheim} reduces to Eq.~\ref{eq:firstpaper} in the case of degenerate energy levels, as shown in~\cite{Aberg11}. In this present article we impose no restrictions on the energy spectra or occupation probabilities, they may take arbitrary form independently of one another.

{\bf The model for work extraction.}\phantomsection{}\addcontentsline{toc}{subsection}{The model for work extraction}
Our work extraction model can be thought of as a game with simple but minimal rules. (It will nevertheless not be trivial to analyse as there is a multitude of different strategies one may choose for the task of work extraction given the initial and final conditions.) The model is inspired by~\cite{AlickiHHH04} and very similar to that used in~\cite{Aberg11}. There are three systems and an  implicit work-extraction agent representing the external experimenter who can control certain parameters. 
As depicted in Figure 1(b) one system is the working medium, another is a heat bath of temperature $T$, and the last is the work reservoir. 
\begin{figure}[htb]
 \centering \includegraphics[width=0.8 \linewidth]{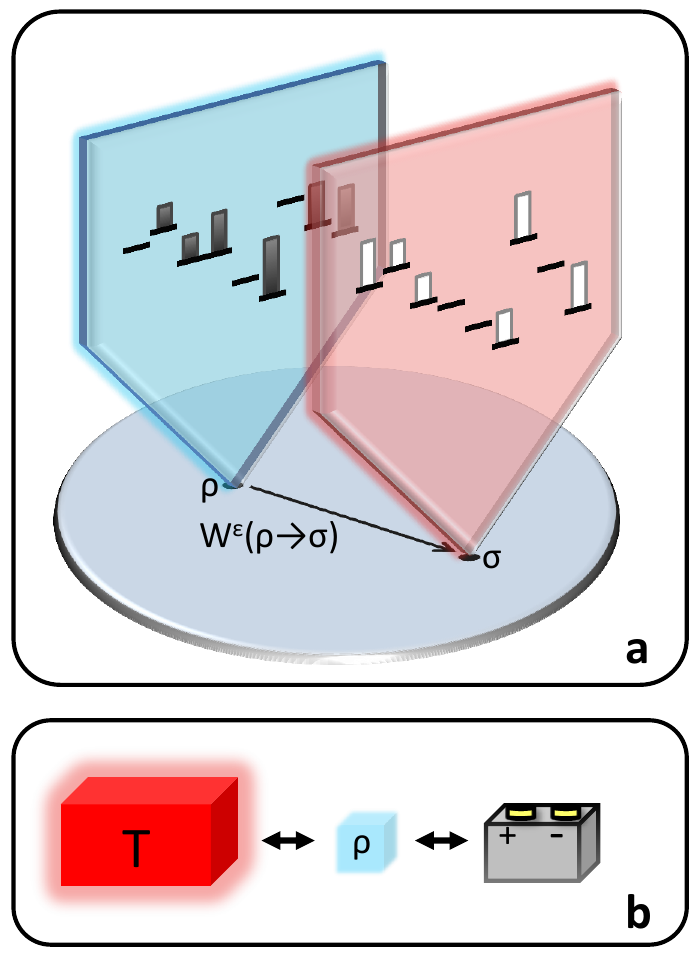}
 \caption{(a) Abstract depiction of the set of states, including the initial state $\rho$ and final state $\sigma$. Each state is associated with a set of energy levels and occupation probabilities. We derive an expression for how much work one can optimally extract with a maximum probability of failure of $\varepsilon$ for any such $\rho$ and $\sigma$. This quantity is called $W^\varepsilon (\rho, H_i \rightarrow \sigma H_f)$. Only in certain limits does it reduce to the standard free energy difference. (b) The generic setup we are considering involves three systems: a heat bath at temperature $T$, a working medium system associated with some initial state $\rho$, and a work reservoir system. One may for instance couple the system to the heat bath and the work reservoir alternately and thereby transfer energy from the heat bath to the work reservoir, at the cost of randomising the working medium system.}
 \label{fig:Levels}
\end{figure}

The initial and final energy spectra $\{E\}$ and $\{F\}$ of the working medium are arbitrary. The initial and final density matrices of the working medium, $\rho$ and $\sigma$, are not assumed to be thermal, they can take any form as long as they are diagonal in the energy basis. This is because we assume, as is non-trivial but standard, that the decoherence time is much faster than the thermalisation time~\cite{AlickiHHH04}. These initial and final conditions are depicted in Figure 1 (a).

One of the two elementary processes the agent can compose to build the full strategy is {\em thermalisation} of the working medium. With thermalisation we mean gradual thermalisation, i.e. we do not mean that the state after the thermalisation process is thermal, but merely that it is nearer to the thermal state than before the process. This is modelled by the probabilities of the energy-levels being transformed by a matrix from the set of stochastic matrices which have the thermal state corresponding to temperature T as the fixed state. This process does not change the Hamiltonian of the working medium. There is by definition no work gain or cost from this process.

The second elementary process is {\em changing the Hamiltonian} of the system through shifting an energy level by some chosen amount $\delta E$. One may for example think of moving a magnet or a charge closer to the system as a way of shifting the levels. This may involve a work gain/cost, because if the system occupies the particular energy eigenstate(s) that gets shifted by $\delta E$ this counts as work done on the system. If the system does not occupy the eigenstate that gets shifted there is no work cost. Importantly, we enforce energy conservation by changing the energy of the work reservoir by the same amount ($\delta E$ if the shifted level is occupied, 0 otherwise). As the system's state is in general not fully known, each Hamiltonian-changing step induces a probability distribution over energy transferred to the work reservoir. For example, if level $i$ only is raised by $\delta E_i$ and the others are stationary the probability of the work reservoir losing $\delta E_i$ of energy is $p_i$, the probability of occupation of level $i$, and the probability of the work reservoir not changing its energy is $1-p_i$. Finally, it is assumed that the experimenter implements Hamiltonian changes without affecting which energy level is occupied. This is justified by the adiabatic theorem which says that it is possible to avoid hopping between levels by shifting them sufficiently slowly. In general this will not be the case but we are interested in fundamental limits and allow the experimenter this level of control.

The agent's choice of how to combine the elementary processes is called its strategy $\mathcal{S}$. Any given strategy will in general generate an associated probability distribution over work costs/gains, i.e. of total energy transfers from/to the work reservoir. When strategy $\mathcal{S}$ is guaranteed to transfer a certain amount of energy up to probability $\varepsilon$ we call this the ($\varepsilon$-) {\em guaranteed work} and denote it by $W_\mathcal{S}^\varepsilon$. In a given realization the strategy $\mathcal{S}$ may then (with a probability bounded by $\varepsilon$) fail to achieve $W_\mathcal{S}^\varepsilon$, otherwise we say the work extraction was successful (in achieving $W_\mathcal{S}^\varepsilon$).

{\bf Relative mixedness gives the optimal guaranteed work.} \phantomsection{}\addcontentsline{toc}{subsection}{Relative mixedness gives the optimal guaranteed work}
In this section we focus on deriving the {\em optimal} amount of work that can be guaranteed to be extracted (up to failure probability $\varepsilon$), writing this as 
$W^\varepsilon(\rho\rightarrow \sigma):=\max_{\mathcal{S}} W_\mathcal{S}^\varepsilon(\rho\rightarrow \sigma)$. The bound we get from these considerations is one of the main results of this paper.

We will show that this is determined by a measure of how much more mixed one state $\rho$ is than another, $\sigma$. We call this the relative mixedness and write it as ${\bf \sf M}(\rho||\sigma)$. As we consider states diagonal in the energy basis, the only relevant information about a state will be its spectrum. For our purposes it will therefore be enough to define the relative mixedness for probability distributions.
\begin{defin}
Consider two probability distributions $\lambda(x)$ and $\mu(x)$ defined over $x\in \mathbbm {R}^{(\geq 0)}$. Let $\lambda(x)\!\downarrow$ and 
$\mu(x)\!\downarrow$ denote these distributions after a (measure-preserving) rearrangement so that they are in descending order. Let the cumulative distribution function associated with a function $\gamma$ be denoted as $\mathcal{F}_\gamma (x):=\int_0^{x} dx' \gamma (x')$. Then the relative mixedness of $\lambda (x)$ and $\mu (x)$ is defined as 
\begin{eqnarray*}
{{\bf \sf M}}(\lambda||\mu):=\max\, m \,\mathrm{s.t.}\, \mathcal{F}_{\lambda\downarrow} \left(\frac{x}{m}\right)\geq \mathcal{F}_{\mu\downarrow} (x)\,\,\,\forall x, 
\end{eqnarray*}
where $m\in \mathbbm {R}$. In words: the relative mixedness of $\lambda$ and $\mu$ is the maximal amount by which one can stretch $\lambda\!\downarrow$ under the condition that its integral upper bounds the integral of $\mu\!\downarrow$ at all points. 
\end{defin}

By the definition of majorisation, if and only if ${\bf \sf M}\geq 1$ does (the spectrum of) $\rho$ majorise $\sigma$, $\rho \succ \sigma$. The actual number ${{\bf \sf M}}$ can thus be viewed as putting a number to how much $\rho$ majorises $\sigma$. 

We shall make use of a powerful insight from~\cite{RuchM76,Ruch75,Mead77}, who were---to our knowledge---the first to note that the decreasing of the von Neumann entropy might not be a sufficient criterion for characterizing thermodynamical processes and they proposed a criterion based on majorization; this insight is also used in~\cite{HorodeckiO11} where they showed this criterion to be necessary and sufficient for a class of quantum operations introduced in \cite{janzing2000thermodynamic}. A relation between majorization and thermodynamics has also been noted in~\cite{janzing2000thermodynamic,janzing2006computer,horodecki2003reversible,AllahverdayanBN04}. The insight bridges a particular gap between information theory and statistical mechanics: the fact that the former does not care about energy. In information theory, the Shannon/von Neumann entropy of a state, $-\sum_i \lambda_i\log \lambda_i$ is independent of the energies of the states involved. As the extractable work should depend on the energy levels involved it follows that it is not expected to be uniquely determined by an entropy.

A key way in which energy enters into statistical mechanics is that in a Gibbs state the probability of any given energy eigenstate with energy $E$ is given by $p_T(E)=\exp(-\frac{E}{kT})/Z$, where $Z$ is the partition function. The insight 
we adapt from~\cite{RuchM76,Ruch75,Mead77} is that we can take this bias into account by what essentially amounts to rescaling the density matrix's eigenvalue distribution by $p_T(E)$.  After the rescaling the occupation probabilities will turn out to uniquely determine our expression for the extractable work.  More specifically, we shall be employing an operation we term {\em Gibbs-rescaling} to the eigenvalue spectrum. Consider states with discrete spectra $\{\lambda_i\}$. We firstly transform the spectrum into the associated step-function. Then we take each block, rescale its height as $\lambda_i\mapsto \lambda_i/ \exp\left(-\frac{ E_i}{kT}\right)$, and its width $l=1\mapsto \exp\left(-\frac{ E_i}{kT}\right)$ such that the area of the new block is $\lambda_i$ as before.
We write this operation applied to a density matrix $\rho$ as  $G^T(\rho)$,  or $G^{(T,H)}(\rho)$ to make the  dependence on the Hamiltonian $H$ explicit.

A way of understanding the Gibbs-rescaling is to think of it as splitting events into finer events in such a way that a Gibbs state becomes a uniform distribution, i.e. higher probability events get split into more fine events than those with lower probability. This fine-graining may even be thought of as physically associated with the number of joint states on the system and the heat-bath, with high probability states associated with more joint states on the system plus environment than low probability states.

Having defined the relative mixedness ${{\bf\sf M}}(.||.)$ and Gibbs-rescaling $G^T(.)$ we can now give the main result. This result states that given that the chosen strategy must take an initial state $\rho$ to a final state $\sigma$ and the initial Hamiltonian $H_i$ to $H_f$, the optimal work that can be guaranteed up to probability $\varepsilon$ to be extracted, $W^\varepsilon (\rho, H_i\rightarrow \sigma, H_f)$, is given by the relative mixedness of the Gibbs-rescaled states.
\begin{thm}\label{th:main}
In the work extraction game defined above, consider an initial density matrix $\rho=\sum_i \lambda_i\ket{e_i}\bra{e_i}$ and final density matrix $\sigma=\sum_j \nu_j\ket{f_j}\bra{f_j}$ with $\{\ket{e_i}\}$, $\{\ket{f_j}\}$ the respective energy eigenstates of $H_i$ and $H_f$. Then for any strategy $\mathcal{S}$, $W^\varepsilon_\mathcal{S} (\rho, H_i \!\rightarrow \! \sigma, H_f)\leq W^\varepsilon (\rho, H_i \!\rightarrow \! \sigma, H_f)$, where 
\begin{eqnarray*} W^\varepsilon (\rho, H_i\rightarrow \sigma, H_f) = kT\ln \left({\bf \sf M}\left(\frac{G^{(T,H_i)}(\rho)}{1-\varepsilon}||G^{(T,H_f)}(\sigma)\right)\right). 
\end{eqnarray*}
 Furthermore an explicit strategy we propose  always saturates this bound, provided that the agent can access a single extra two-level system (the  catalyst system) which is fixed to be in one of its energy eigenstates with $\ket{\xi}\bra{\xi}$ both initially and finally, i.e. $\rho=...\otimes \ket{\xi}\bra{\xi}$ and $\sigma=... \otimes \ket{\xi}\bra{\xi}$ with the same initial and final Hamiltonian on the catalyst.
\end{thm}
\noindent Here we give the main arguments for the theorem, a full proof is given in the appendix.

The first claim concerning the relative mixedness expression on the RHS being an upper bound is arrived at from the following line of reasoning. There are two elementary processes and each have the effect of making the state more (or at least not less) mixed according to the relative mixedness measure. Work extraction, by definition, only occurs during a change of the Hamiltonian. In this case the optimal is to only move occupied levels, for which the energy gain is given precisely by $kT\ln \left({\bf \sf M}\left(G^{(T,H_i)}(\rho)||G^{(T,H_f)}(\sigma)\right)\right)$ (see the Appendix).

The second claim concerns a universal strategy that we formulate. To illustrate it we now describe a very simple instance: the case of Landauer's bit reset with certainty ($\varepsilon =0$). Here there is a qubit associated with two energy levels $E_1$ and $E_2$ with $H=E_1\ket{1}\bra{1}+E_2\ket{2}\bra{2}$. We demand $E_1=E_2=0$ at the beginning and at the end, $\rho_i=1/2\ket{1}\bra{1}+1/2\ket{2}\bra{2}$, $\rho_f=\ket{1}\bra{1}$. The change in the state is why this is called `bit reset' (it is often called, ambiguously, bit {\em erasure}). 
Our universal strategy reduces in this simple case to the following: (i) lift both energy levels up by $\Delta E=kT\ln2$. This costs $kT\ln2$ of work with probability 1, (ii) split the levels quasistatically and isothermally such that $E_1\rightarrow 0$ and $E_2\rightarrow \infty$. In this step the Gibbs rescaled distributions are not changed, they are all `Gibbs-equivalent'.  This level splitting actually costs 0 work with probability 1. This can be seen by making use of the powerful Mc Diarmid's inequality~\cite{McDiarmid}. The key step is to argue that {\em lifting an individual level quasistatically and isothermally gives a probability distribution over work that has arbitrarily small spread around the average}. This can be shown by considering a series of discrete lifts of the same size $\Delta E$ with the work cost a random variable for each one. The work cost of one step is independent of that of any other step, because the state is by assumption thermal before each lift (as follows from the process being isothermal and quasistatic). Mc Diarmid's inequality states: {\em Let $X_1$, $X_2$...$X_n$ be independent random variables all taking values in the same set. Call the realised value of $X_i$ $x_i$. Further, let $f(x_1,x_2...)$ be a real-valued function with the property that changing one of the $x_i$ only can at most change $f$ by $c_i$. Then for all $\epsilon >0$, $Pr(|f-\mathbbm{E}(f)|\geq\epsilon)\leq exp\left(\frac{-2\epsilon^2}{\sum_{i=1}^{n}{c_i}^2}\right)$}. Letting the random variables be the energy transferred to the work reservoir in each step, and $f$ be the total energy transferred, one can with a little effort show that there is indeed no deviation from the mean. We note that~\cite{Aberg11} contains alternative techniques for showing concentration around the mean and that, moreover, in the a priori different setting used in~\cite{HorodeckiO11} what amounts to Gibbs-equivalent transforms at zero work cost are also possible. (iii) Finally the system is decoupled from the heat bath and the empty level 2 is moved down to $E_2=0$ (without any work cost/gain), completing the process.

It is an interesting question how one could generalize our theorem. In the more general case of off-diagonal terms in the energy eigenbasis, one expects entanglement to arise between the work reservoir and the working medium system during the work extraction steps and it is subtle how to define work as the energy of the work-reservoir is  not well-defined. One analytically clean approach is to allow decoherence in the systems energy basis a free operation for the experimenter, as in~\cite{HorodeckiO11}. Then the corresponding decohered state can be inserted into the above expression, implying that the relative mixedness of the decohered state relative to the final state gives a lower bound on the extractable work in the case of off-diagonal terms.

Several existing results are recovered as special cases of Theorem 1. Eq.~\ref{eq:AbergOppenheim} above (from~\cite{Aberg11,HorodeckiO11}) and accordingly Eq.~\ref{eq:firstpaper} (from~\cite{DahlstenRRV11}) are special cases of our main result---see the supplementary information (we reiterate that \cite{HorodeckiO11} uses an a priori distinct set-up and note that the work referred to there is 'deterministic' work associated with deterministic energy transfers to a constantly pure work reservoir and is a priori distinct from the 'guaranteed' work considered here). Eq.~\ref{eq:AbergOppenheim} corresponds to the case where the final state $\rho_T$ is demanded to have the same eigenspectrum and be a Gibbs state ($\rho_T=\sum p_T(E_i)\ket{e_i}\bra{e_i})$).   If the initial and final states are both thermal with associated partition functions $Z_i$ and $Z_f$ the expression reduces to $kT\ln\frac{Z_f}{Z_i}$ (as is consistent with~\cite{HorodeckiO11,Aberg11}). To our knowledge our paper is the first to give an expression for the optimal work (guaranteed) to be extractable from a general
energy-diagonal state to another, with changing Hamiltonians and
possibly non-zero risk. In~\cite{HorodeckiO11} they also consider how
one can calculate the work that can be extracted with arbitrary initial
and final Hamiltonian, with either the initial or the final state being
thermal and showing how the thermo-majorization condition describes the
zero-risk, deterministic work for arbitrary energy-diagonal initial and final states.

{\bf Generalised Laws of Thermodynamics in terms of relative mixedness.}\phantomsection{}\addcontentsline{toc}{subsection}{Generalised Laws of Thermodynamics in terms of relative mixedness}
As the laws of thermodynamics are centered around the notions of energy, work and entropy, these laws should according to our argument also be formulated in terms of relative mixedness for them to be more suitable beyond the asymptotic i.i.d.\ regime.

{\em $0^{\text{th}}$  law:}
The $0^{\text{th}}$ law can be stated as: {\em  
There exists for every thermodynamic system in equilibrium a property called temperature. Equality of temperature is a necessary and sufficient condition for thermal equilibrium.} This also holds after our generalisation. In particular we are still assuming heat baths that take the working medium closer to a Gibbs thermal state upon interaction. 

{\em  First law:}
The first law can be viewed as both asserting the conservation of energy as well as stating that it can be divided into two parts, work and heat, which are normally defined in the description accompanying the first law equation: $dU=dQ-dW$. $U=\text{tr}(\rho H)$ is the expected internal energy of the working medium with Hamiltonian $H$, $Q$ is `heat' and $W$ `work'.
The associated physical setting is that there is a working medium system which can either exchange energy with another system in a thermal state dubbed a heat bath, or with a work reservoir system normally implicitly assumed to be in some energy eigenstate of its own Hamiltonian. Exchanges of energy with the heat bath are dubbed heat and those with the work reservoir work. This essentially carries over into our approach but with some important subtleties. We assume energy conservation (in every single extraction), as well as allowing for interactions with a heat bath and a work reservoir. Thus the following is respected when the actual energy of the system $E_{\text{sys}}$ changes: $dE_{\text{sys}}=-dE_{\text{bath}}-dE_{\text{reservoir}}$. 
We, more subtly, break  $dE_{\text{reservoir}}$ into two parts: $dE_{\text{reservoir}}=dW^{\varepsilon}_\mathcal{S}+dE_{\text{extra}}$. There is the energy transfer which is predictable (up to $\varepsilon$ probability of failure) in that it corresponds to $dW^{\varepsilon}_\mathcal{S}(\rho \rightarrow \sigma)$ for the infinitessimal state change $\rho \rightarrow \sigma$ using strategy $\mathcal{S}$. We view anything beyond that, given by $dE_{\text{extra}}$, as heat (even though this energy flows into the work reservoir at first). The idea behind this is that only predicted energy transfer should count as work. One may for example imagine buckets lifting water out of a mine up to a certain height (or as a quantum example an electron excited into the conduction band). The height at which the buckets are tipped into a reservoir is specified in advance. If they go higher than this, the extra potential energy will be transferred to other degrees of freedom associated with the reservoir system, e.g.\ into movement of the water (or heating of the semi-conductor). We may express the following first law for this approach:\\
{\em In any given extraction, with probability p$\geq 1-\varepsilon$}
\begin{equation}
dE_{\text{sys}}=-dE_{\text{bath}}-dW^{\varepsilon}_\mathcal{S}-dE_{\text{extra}}\equiv dQ-dW^{\varepsilon}_\mathcal{S},
\end{equation}

{\em Second law:}
Consider next the so-called Kelvin statement of the second law: {\em No process is possible in which the sole result is the absorption of heat from a reservoir and its complete conversion into work.} This does not say anything about processes with a non-zero probability of failure. We show in the appendix that for given states of the working medium A and B respectively, ${W^\varepsilon(A\rightarrow B)}+{W^\varepsilon(B \rightarrow A)}\leq {W^{2\varepsilon}(A \rightarrow A)}.$ We call this the {\em triangle inequality}. It implies together with the main theorem that all strategies in our game respect the following generalisation of Kelvin's second law:
\begin{equation}
\sum_{i=0}^{m-1} W_{\mathcal{S}_i}^\varepsilon \left( A_i\rightarrow A_{i+1} \right )  \leq  W^{m\varepsilon}(A \rightarrow A) \text{ if }A_m=A_1, 
\end{equation}
 where $\mathcal{S}_i$ is the choice of strategy in the i-th step of the cycle. Note that $W^0(A \rightarrow A)=0$ (see main theorem), implying that deterministically no work can be extracted in such a cycle. One may still gain work in a single cycle at the cost of having $\varepsilon >0$ for one or more of the steps. 
 
The second law is also closely related to entropy increasing with time and one may wonder what the corresponding generalisation of the statement is. A particular standard expression is that
\begin{equation}
\label{eq:entropyincrease}
\Delta \left(S-\beta \langle E \rangle \right) \geq 0,
\end{equation}
 where $S$ and $\langle E \rangle$ are the von Neumann entropy and expected energy of a system interacting with a heat-bath with inverse temperature $\beta$. ($\Delta$ indicates the change in these values during the interaction.)  
This actually still holds in our more general model; we show this in the supplementary information. However, crucially, Eq.\ref{eq:entropyincrease} is not {\em sufficient} to guarantee that an evolution $\rho \rightarrow \rho'$ is realizable through an interaction with a heat bath. Instead it should be replaced by the statement that a state change $\rho \rightarrow \rho'$ due to a thermalisation with a heat-bath at temperature $T$ is possible if and only if
\begin{equation}
\label{eq:altentropyincrease}
W^0(\rho, H \rightarrow \rho',H)\geq 0. 
\end{equation}
This is significant as there are processes that respect Eq.~\ref{eq:entropyincrease} but violate Eq.~\ref{eq:altentropyincrease}. A simple example is to consider degenerate energy levels, so that $\Delta \langle E \rangle=0$, and  three levels with probabilities $(1/2\,\,\,1/2\,\,\,0)^T \rightarrow (2/3\,\,\,1/6\,\,\,1/6)^T$. Then $\Delta S\approx 0.25$ but $W^0$ is negative. Strikingly, such evolutions enable the deterministic violation of Kelvin's second law (if the evolution is stochastic---see  supplementary information).

The inequivalence of entropy and majorisation has been noted previously in the context of the second law~\cite{RuchM76,Ruch75}. Presumably this has not received more attention to date because in the von Neumann regime this inequivalence disappears. More precisely, if we consider a tensor product of $n$ identical states each with von Neumann entropy $S$ and let $n\rightarrow \infty$, then with asymptotically small error we may approximate the spectrum as a uniform probability distribution on the set [0,$2^{-nS}$]. For such distributions the partial orders induced by $S$ and majorisation respectively coincide.

We finally make a remark on the mathematical structure that emerges here. We note that the extractable work is no longer a function of state, whereas in standard statistical mechanics the optimal extractable work between two states is given by $\delta F_{12}=F_2-F_1$ with $F=U-TS$. Here one must consider the extractable work between two states, assigning a free energy as a state function is not possible. It is not even optimal to go via thermal states in general, i.e., there exist cases where $W^{\varepsilon}(\rho\rightarrow \sigma_T)+W^{\varepsilon}(\sigma_T\rightarrow \sigma)< W^{\varepsilon}(\rho\rightarrow \sigma)$. 

Very recently it has been argued that our generalized formulation of the second law should be replaced with a slightly weaker condition~\cite{brandao2013second}. As this appeared after our paper on the arXiv we defer discussion of the relation between these papers to later work. In between this paper appearing on the arXiv and being published several other related, interesting and relevant contributions have appeared, including~\cite{2012arXiv1211.1037F,gour2013resource,lostaglio_2014}.

{\bf Relative mixedness as entanglement measure.} \phantomsection{}\addcontentsline{toc}{subsection}{Relative mixedness as entanglement measure}The structures of entanglement theory and thermodynamics are closely linked and often considered in connection with one another, see e.g.~\cite{PlenioV98}. We now consider the implications of our results for entanglement theory. This section demonstrates that relative mixedness is natural to use in quantum information theory also outside of thermodynamical contexts.
It is customary to quantify entanglement via entropy, in particular the standard measure of entanglement of a bipartite pure state $\rho_{AB}$ is the von Neumann entropy of the reduced state, $S(\rho_A)=S(\rho_B)$. This is called the {\em entanglement entropy}. However there is good reason to think that, as we have argued in the case of statistical mechanics, entropy should be replaced with relative mixedness also in the context of entanglement theory. We propose a notion of relative entanglement between two states $\rho_{AB}$ and $\sigma_{AB}$ which is quantified as the (logarithmic) relative mixedness of the reduced states: $\log_2{{\bf \sf M}}(\sigma_A || \rho_{A})$.

This has the following appealing operational meaning. Consider the Bell state $\ket{\phi^+}_{AB}:=\frac{1}{\sqrt{2}}(\ket{0}_A\ket{0}_B+\ket{1}_A\ket{1}_B)$. Consider two arbitrary finite-dimensional bipartite pure states $\rho_{AB}$ and $\sigma_{AB}$. How many such Bell pairs are needed to transform $\rho_{AB}$ to $\sigma_{AB}$? More specifically, for what condition on $n_i$ and $n_f$ is the LOCC (Local Operations and Classical Communication) conversion $\rho_{AB}\otimes (\ket{\phi^+}\bra{\phi^+}_{AB})^{\otimes n_i} \rightarrow \sigma_{AB}\otimes (\ket{\phi^+}\bra{\phi^+}_{AB})^{\otimes n_f}$ possible? The answer is that this is possible iff 
\begin{eqnarray*}
n_f-n_i \leq \log_2 {{\bf\sf M(\sigma_A||\rho_A)}}.
\end{eqnarray*}
(We prove this in the Appendix, making heavy use of the results of~\cite{Nielsen99} and the setting of~\cite{BuscemiD11}).

As a very simple example, for $\ket{\psi}=\alpha\ket{00}+\beta\ket{11}$  (and $\alpha \geq \beta$)  and $\ket{\phi}=\ket{\phi^+}_{AB}$ one finds $\log_2 {{\bf\sf M}}(Tr_B\ket{\psi}\bra{\psi}||Tr_B\ket{\phi}\bra{\phi})=\log_2 \left(2\|\alpha\|^2\right)$. This takes values between 1  
($\alpha =1$)  
and 0 ($\alpha=\frac{1}{\sqrt{2}}$).

\ack
We gratefully acknowledge discussions with J. Aaberg, J. Baez, B. Fong, P. Perinotti, J. Vicary, M. Horodecki and J. Oppenheim, as well as support from the National Research Foundation (Singapore), the Ministry of Education (Singapore), the Swiss National Science Foundation (grant No. 200020-135048), the Swiss National Centre of Competence in Research QSIT, the European COST Action on Quantum Thermodynamics, the EU Integrating Project SIQS, the European Research Council (grant No. 258932) and the EU collaborative project TherMiQ (Grant agreement No. 618074). DE is grateful for the hospitality of the Clarendon Laboratory, University of Oxford, whilst undertaking part of this work. This research was partly carried out in connection with DE's Master's thesis at ETH Zurich.

\section*{References}
\phantomsection{}\addcontentsline{toc}{section}{Bibliography}
\bibliography{SIWErefs}
\bibliographystyle{jphysicsB}
\phantomsection{}\addcontentsline{toc}{part}{Appendix}
\appendix
\setcounter{section}{0}
\include{tech_setup}
\include{tech_p5_4d}

\include{tech_achievable26_9_13}
\include{tech_2ndlaw}
\include{minrelative4pages}

\include{entanglement_measure}

\end{document}

%% file: tech_setup.tex
\section*{Appendix}
The appendix is structured in the following manner. A: The work extraction game, B: Upper bounding the extractable work, C: The universal strategy that achieves the bound, D: Implications for the second law, and E-G: Properties of the relative mixedness.

\section{The work extraction game}\label{sec:game}
In this section we define the setting more carefully, and derive certain lemmas which shall be needed for the later sections. 

\subsection{Combining energy and occupation probabilities into one distribution: Gibbs rescaling}
There are two central pieces of information about the system, the energy eigenvalues, and their occupation probabilities. 
We shall find it very powerful to follow~\cite{RuchM76,Ruch75,Mead77} and combine them into one object, the Gibbs-rescaled distribution. 

Consider states with discrete spectra $\{\lambda_i\}$. We firstly transform the spectrum into the associated step-function. Then we take each block, rescale its height as $\lambda_i\mapsto \lambda_i/ \exp\left(-\frac{ E_i}{kT}\right)$, and its width $l=1\mapsto \exp\left(-\frac{ E_i}{kT}\right)$ such that the area of the new block is $\lambda_i$ as before.
We write this operation applied to a density matrix $\rho$ as  $G^T(\rho)$. It is depicted in figure~\ref{fig:gibbs}.  Gibbs-rescaling can, as will prove useful in later proofs, be written out in the language of continuous functions in the following manner:
\begin{defin}[Gibbs rescaling]\label{def:gibbsr}
Consider a density matrix $\rho=\sum_{i=1}^n \lambda_i\ket{e_i}\bra{e_i}$ with eigenvalues $\{\lambda_i\}_{i=1}^n$ and take the energy eigenstates of the system to be $\{\ket{e_i}\}_{i=1}^n$ with energies $\{E_i\}_{i=1}^n$ respectively. There is an associated step function for the spectrum, $\lambda(xn)=\lambda_{\lceil{xn}\rceil}$ where $x \in (0,1]$. Similarly there is an energy step function $E(xn)=E_{\lceil{xn}\rceil}$ where $x \in (0,1]$. The Gibbs rescaling associated with temperature $T$ combines $\lambda(xn)$ and $E(xn)$ to a new function $G^T(y)$   implicitly defined by  
\[G^T\left(\int\limits_0^x e^{-\frac{E_{\left\lceil zn \right\rceil}}{kT}}  \md z \right)
  =
    \frac{\lambda_{\left\lceil xn\right\rceil}}{e^{-\frac{E_{\left\lceil  xn \right\rceil}}{kT}}}.
 \] 
\end{defin}

It follows that $G^T (y)$ is defined on $(0,Z]$, with $Z=\sum_{j=1}^n \exp\left(-\frac{E_{j}}{kT}\right)$ the partition function. Moreover $G^T (y)$ is a probability distribution satisfying $\int_0^Z G^T(y)dy=1$. 
\begin{figure}
 \centering
 \includegraphics{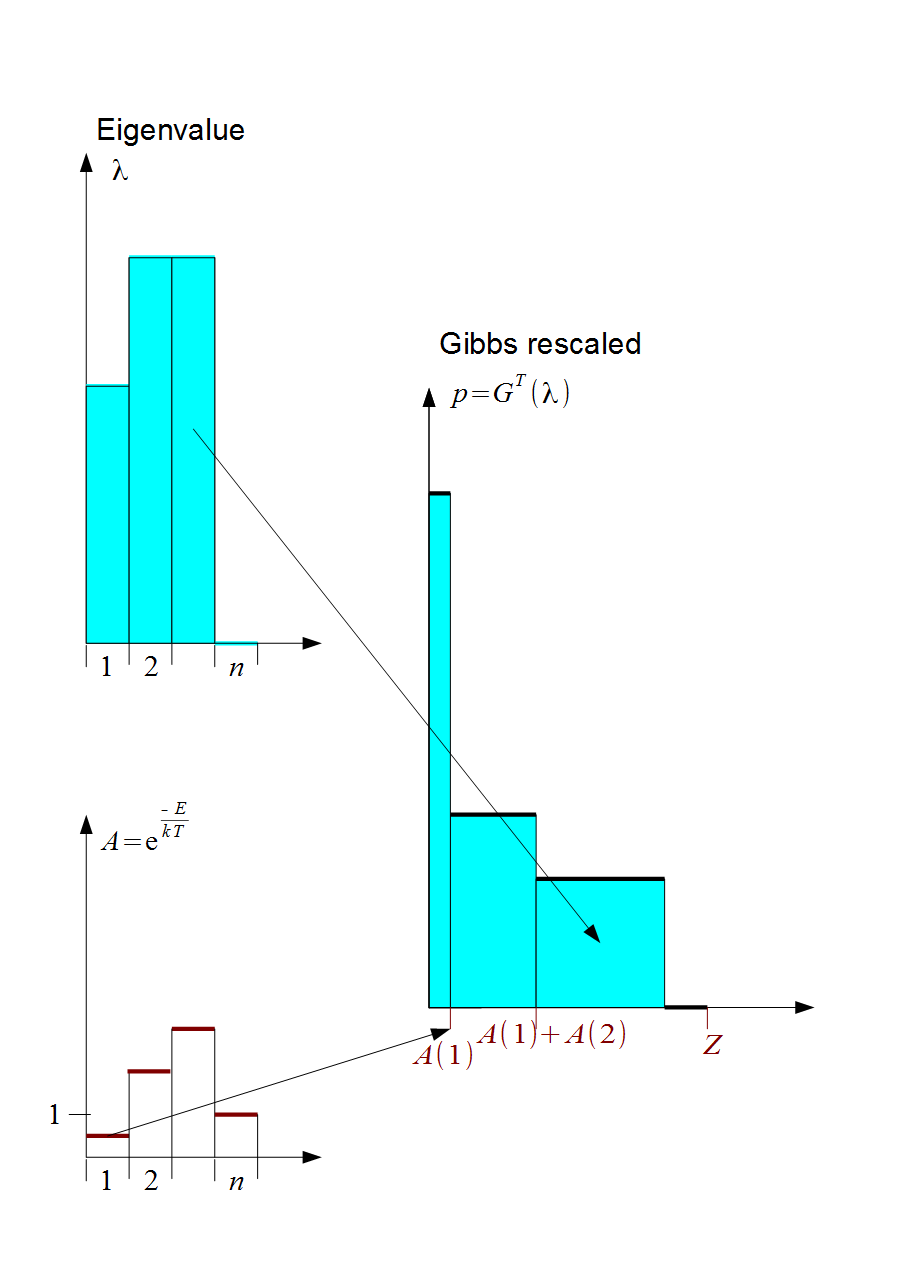}
 \caption{Gibbs rescaling: the width of each block $k$ corresponding to the level $k$ after rescaling is given by $A(k)=\exp(-E(k)/kT)$, while its height is $\lambda(k) / A(k)$ so that its area is $\lambda(k)$, where $\lambda(k)$ is the occupation probability of the level and $E(k)$ its energy eigenvalue.}
 \label{fig:gibbs}
\end{figure}

\subsection{Thermalisations}
We now turn to how interactions with the heat bath, thermalizations, act on the state of the system. Roughly speaking these take the density matrix closer to the associated Gibbs state, similar statements can be found in~\cite{RuchM76,Ruch75,Mead77} (especially see section 4 of~\cite{RuchM76}, where also a different argument is given for the result below concerning thermalizations). As already mentioned the thermalization is taken to only change occupation probabilities and not energy eigenvalues. We take the thermalisation to act as a stochastic process on the energy eigenstates, in that the probability of occupying a given energy state, $P(i)$, becomes $P'(i)=\sum_{j}P(j\rightarrow i)P(j)$ where the summation is over all eigenstates, $P(j\rightarrow i)$ is a transition probability, and $P(j)$ an occupation probability (before the interaction with the heat bath). This can equivalently be written as $\vec{P'}=B\vec{P}$ where $B$ is a stochastic matrix (entries are probabilities and columns sum to 1). 

Not every stochastic matrix $B$ is allowed however. The Gibbs state (associated with temperature $T$) is taken to be invariant under a thermalisation. Consider the implications firstly for the fully degenerate case of all energies being the same. In this case the Gibbs state is the uniform distribution. The only stochastic matrices that leave the uniform distribution invariant are bistochastic ones (rows also sum to 1). Thus in the fully degenerate case $B$ must be bistochastic. We see no reason to impose further restrictions, so any such $B$ is allowed. 

Consider secondly the non-degenerate case. Here it is again convenient to use the Gibbs rescaled distribution. Note that the Gibbs state becomes uniform after the Gibbs rescaling. Thus one may hope that a thermalisation, i.e.\ a Gibbs state preserving stochastic matrix on the occupation probabilities, acts as a {\em bi}-stochastic matrix on the Gibbs-rescaled distribution, and we now show that is indeed the case.

Before considering the general case, we look at a simple example of a two-level system.\\
Let B be the stochastic matrix\footnote{stochastic matrices have entries in $[0,1]$ with columns summing to $1$, therefore they map probability vectors to probability vectors.} defined by the transition-probabilities, i.e:
\[\begin{pmatrix}
P(1) \\
P(2) \\
\end{pmatrix} \rightarrow \begin{pmatrix}
P'(1) \\
P'(2) \\
\end{pmatrix} =
\begin{pmatrix}
p_{(1\rightarrow 1)} & p_{(2\rightarrow 1)} \\
p_{(1\rightarrow 2)} & p_{(2\rightarrow 2)} \\
\end{pmatrix}
\begin{pmatrix}
P(1) \\
P(2) \\
\end{pmatrix}
= B \begin{pmatrix}
P(1) \\
P(2) \\
\end{pmatrix}\]

The stochastic matrix should leave the thermal state invariant: 
\[\begin{pmatrix}
e(1)/Z \\
e(2)/Z \\
\end{pmatrix} \rightarrow \begin{pmatrix}
e'(1)/Z \\
e'(2)/Z \\
\end{pmatrix} =
B \begin{pmatrix}
e(1)/Z \\
e(2)/Z \\
\end{pmatrix}
=\begin{pmatrix}
e(1)/Z \\
e(2)/Z \\
\end{pmatrix}
\]
where $e^{(')}(i)/Z=\exp(-E^{(')}(i)/(kT))/Z$ is the Gibbs state (which should be invariant as the energy does not change) and $Z=e(1)+e(2)$.

Look at what happens with $e(1)= 2$ and $e(2)=1$. For the Gibbs rescaling this means that $P(1)\rightarrow P(1)/2$ on the length $2$ and $P(2)\rightarrow P(2)$ on the length $1$. We can split the first level into two parts (in our mind) and consider new levels $(P(1_1),P(1_2),P(2_1))=P(1)/2,P(1)/2,P(2)/1)$ all having the same length after Gibbs rescaling. For the thermal state this means:

\[\begin{pmatrix}
2/3 \\
1/3 \\
\end{pmatrix} \rightarrow \begin{pmatrix}
1/3 \\
1/3 \\
1/3 \\
\end{pmatrix}
\]
The transition matrix becomes:
\[\begin{pmatrix}
p_{(1\rightarrow 1)} & p_{(2\rightarrow 1)} \\
p_{(1\rightarrow 2)} & p_{(2\rightarrow 2)} \\
\end{pmatrix}
\rightarrow
 \begin{pmatrix}
p_{(1\rightarrow 1)}/2 & p_{(1\rightarrow 1)}/2 & p_{(2\rightarrow 1)}/2\\
p_{(1\rightarrow 1)}/2 & p_{(1\rightarrow 1)}/2 & p_{(2\rightarrow 1)}/2 \\
p_{(1\rightarrow 2)} & p_{(1\rightarrow 2)} & p_{(2\rightarrow 2)} \\
\end{pmatrix}
\]
which is still stochastic, because the initial matrix was. Since the thermal state has to be invariant under the action of this matrix and the thermal state in this case is proportional to the identity, it is straightforward to check that the matrix has to be bistochastic (rows and columns sum to 1).

For the general case   consider dividing the Gibbs-rescaled distribution into fine blocks such that all fine blocks have the same width $w$. Let $N$ be the number of fine blocks. (As the maximum support is given by the partition function $Z$ we have $w=Z/N$). Let $N_k$ be the number of fine grained blocks associated with level $k$, such that $\sum_{k=1}^n N_k=N$. Each energy level is associated with one block only labelled by $k$. Each $l$-th fine block is associated with a level $k_l$. 

Fine blocks associated with the same energy level $k$ must all have the same height, given by $P(k_l)/e(k_l)$   (where $e(k_l) = \exp(-E(k_l)/kT)= w N_{k_l} = Z N_{k_l}/N$ is the total width of the level $k_l$ after Gibbs rescaling. See the comment after the definition of Gibbs rescaling \ref{def:gibbsr}). 
Let $\vec{f}$ contain the $N$ heights of the fine blocks, with $P(k_l)/e(k_l)=P(k_l)N/(Z{N_{k_l}})$ as its $l$-th entry.
Now when the occupation probabilities transform under $B$, $\vec{f}$ undergoes an associated transform. We will argue it is given by a matrix $F$ whose entry in the l-th row and m-th column is given by 
\begin{equation}
F_{lm}:=\frac{B_{k_l k_m}}{N_{k_l}}.
\end{equation}
To see this note firstly that $P'_i=\sum_{j}B_{ij}P_j=\sum_{j}\frac{N_j}{N_j}B_{ij}P_j$, and recall that $f'_l=P'_{k_l}N/(Z{N_{k_l}})$. Thus
\begin{eqnarray}
f'_l &=&\sum_{j=1}^n\frac{N_j}{N_j}B_{k_l j}P_j\frac{N}{ZN_{k_l}}\\
&=&\sum_{j=1}^n\frac{N_j B_{k_l j} }{N_{k_l}}\frac{P_jN}{ZN_{j}}\\
&=&\sum_{j=1}^n\sum_{m | k_m=j }\frac{B_{k_l j}}{N_{k_l}}f_m\\
&=&\sum_{m=1}^N\frac{B_{k_l k_m}}{N_{k_l}}f_m.
\end{eqnarray}
As $B_{ij}$ and $N$ are nonnegative real numbers $F$ has non-negative real entries only. To see that the columns sum to 1 so that $F$ is a stochastic matrix, note that the column sums are the same as for $B$ which is stochastic. 
Moreover as $B$ must leave the Gibbs state invariant, and this is a uniform distribution after the Gibbs rescaling, $F$  must leave the uniform distribution (or anything proportional to it) invariant. Then for any row $i$: $\sum_j F_{ij} (1/N)= 1/N$ so each row of F must sum to 1. Therefore $F$ is a bistochastic matrix. Note that $F$ is additionally restricted, through being defined via $B$, to keep the heights of fine blocks the same whenever these are associated with the same level.

Accordingly we define interactions with the heat-baths, thermalizations, to act in the following way on the system.
\begin{defin}[Thermalization]\label{def:therm}
A thermalization leaves the energy eigenvalues invariant. It acts on the occupation probabilities, i.e. the eigenvalues of the density matrix, as a stochastic matrix. This stochastic matrix leaves the Gibbs state $\exp (\beta H)/Z$ invariant. It follows from this definition and the definition of the Gibbs-rescaled distribution that a thermalisation acts on the Gibbs-rescaled distribution as a bistochastic matrix.
\end{defin}

\subsection{Work extractions}

The second elementary process is {\em changing the Hamiltonian} of the system through shifting a set of energy levels by some predetermined amount $\Delta E(j)$, where $j$ labels the j-th work extraction. This may involve a work gain/cost, because if the system occupies one of the energy eigenstates that get shifted by $\Delta E(j)$ this counts as work done on the system and we write $W_j=\Delta E(j)$. It is assumed that this entails an energy transfer of $\Delta E(j)$ to the work reservoir system, so that energy is conserved. If the system does not occupy the eigenstate that gets shifted there is no work cost, $W_j=0$. To reduce the notation later on we will also find it convenient to define the `logarithmic' work $w^j$ s.t.  $W_j:=kT\ln w^j$ (or equivalently
$w^j:=\exp \left( W_j/kT\right )$). 

There is thus for each elementary work extraction a probability distribution over work transfer, with two elements, [$p(W_j=0)$, $p(W_j=\Delta E(j))$]. A sequence of work extractions generates a randomly picked sequence of energy transfers to the work reservoir by, e.g. $\{$0, 0, $\Delta E(3)$, 0, $\Delta E(5)$...$\}$. There is an associated vector of 0's and 1's where a 1 as the j-th entry indicates that there was indeed a work transfer of $\Delta E(j)$ in the j-th step. We call this latter vector $\vec{s}$, and the j-th entry thereof $s_j$. $s_j=0$ means that the levels shifted in work extraction step $j$ were not occupied, and $s_j=1$ means that they were. 

From the perspective of someone who learns $s_j$, the occupation probabilities $\{\lambda_i\}$ change. If $s_j=1$ one projects the state $\rho$ with projector $\Pi_{\text{shifted}}$ onto the set of levels shifted so that the new state is 
\begin{eqnarray*}
\frac{\Pi_{\text{shifted}}\sum_i \lambda_i \ket{i}\bra{i}\Pi_{\text{shifted}}}{\tr(\rho \Pi_{\text{shifted}})}.
\end{eqnarray*}
If instead $s_j=0$ one replaces the projector with one onto the levels that were not shifted. 

We accordingly represent a work extraction in the following manner:
\begin{defin}[Work extraction]\label{def:wext}

We define a work extraction on the first $l$ levels, which are all to get shifted in energy by $\varDelta E(j) = -kT \ln\left(w^j_{\vec{s}|s_j=1}\right)$, while the remaining levels are untouched as follows. Letting $\Theta_{U}(y)$ denote the function that is 1 if $y \in U$ and else 0, the new occupation probabilities and energies are given by:
\begin{itemize}
\item In the case when $s_j=1$ (state of the system is found to be in the levels $(1, \ldots, l)$):
 \[\lambda_{\vec{s}|s_j=1}^j\left( k \right) =  \Theta_{\{1,\ldots,l\}}(k) \frac{\lambda_{\vec{s}}^{j-1}\left( k \right)}{\eta_{\vec{s}|s_j=1}^{j}}\]
 where $\eta_{\vec{s}|s_j=1}^{j} = \sum\limits_{i=1}^l \lambda_{\vec{s}}^{j-1}\left( i \right) $. In this case there is an energy transfer to the reservoir given, in terms of the logarithmic work, by $w^j_{\vec{s}|s_j=1} =\exp (\Delta E(j)/kT)$.

\item In the case when $s_j=0$ (state of the system is {\em not} found to be in the levels $(1, \ldots, l)$): 
\[\lambda_{\vec{s}|s_j=0}^j\left( k \right) = \Theta_{\{l+1,\ldots,n\}}(k) \frac{\lambda_{\vec{s}}^{j-1}\left( k \right)}{\eta_{\vec{s}|s_j=0}^{j}}\]
 where $\eta_{\vec{s}|s_j=0}^{j} = \sum\limits_{i=l+1}^n \lambda_{\vec{s}}^{j-1}\left( i \right) $. In this case there is no energy transfer to the work reservoir, i.e. $w^j_{\vec{s}|s_j=0}=1$.
\end{itemize}
\end{defin}
\begin{figure}
 \centering
\includegraphics{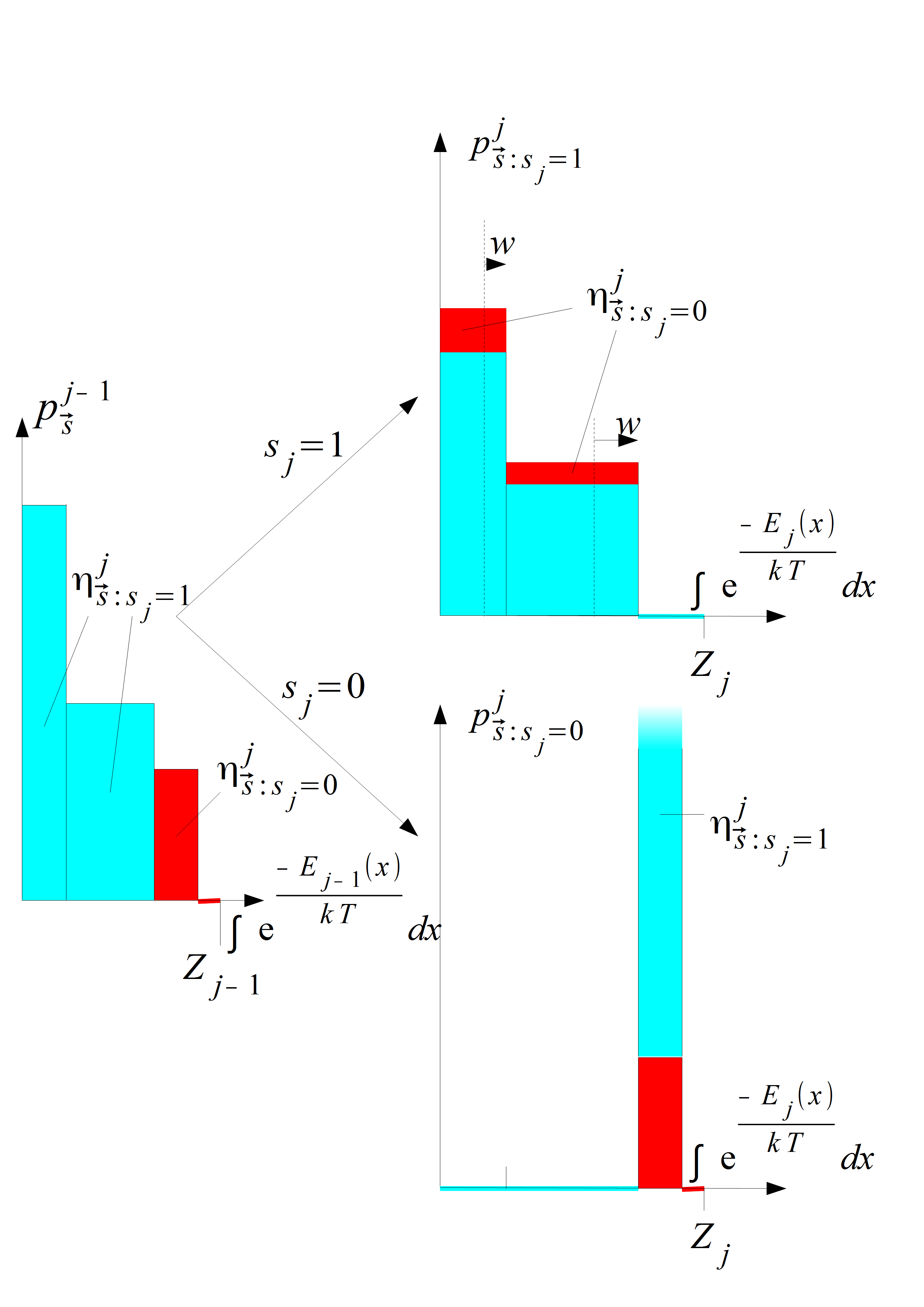}
 \caption{Work extraction: The action of the work extraction on the Gibbs rescaled probability distribution can be seen as a stretching by $w$ of the part from which one tries to extract the work $kT\ln(w)$, followed by a projection onto either the levels from which one tried to extract work (case $s_j=1$) or the rest (case $s_j=0$) followed by a renormalization.}
\label{fig:workextraction} 
\end{figure}

This next lemma considers how the work extraction in the preceding definition acts on the {\em Gibbs-rescaled} distribution. This is also depicted in Figure \ref{fig:workextraction}.

\begin{lem}\label{lem:wext}
Let the levels $\{1,\ldots, l\}$ be used for work extraction as in the above definition.
Let $a \in \mathbb{R}$ be the combined width of the blocks of the Gibbs-rescaled distribution corresponding to the levels $\{1,\ldots, l\}$, i.e. $a=\sum_{i=1}^le^{\frac{-E_i^{j-1}}{kT}}$.
Let $x \in (0,Z_j]$ (with $Z_j$ the partition function after step $j$).

Then following a work extraction in step $j$, the resulting Gibbs rescaled probability distribution, conditioned on the previous steps on path $\vec{s}$, is given by the following.
 In the case where $s_j=1$:
 \[p_{\vec{s}|s_j=1}^j(x)
   =
   \Theta_{(0, aw^j_{\vec{s}|s_j=1}]}(x)
   \frac
     {p_{\vec{s}}^{j-1}
      \left(
	  \frac{x}{w^j_{\vec{s}|s_j=1}}  
      \right)
     }{w^j_{\vec{s}|s_j=1} \eta^j_{\vec{s}|s_j=1}
     }
 \]
In the case where $s_j=0$:
 \[p_{\vec{s}|s_j=0}^j(x)
   =
   \Theta_{(a w^j_{\vec{s}|s_j=1},Z_j]}(x)
   \frac
    {p_{\vec{s}}^{j-1}
      \left(x - a w^j_{\vec{s}|s_j=1} + a\right)
    }{ \eta^j_{\vec{s}|s_j=0}
    }
\]
\end{lem}
\begin{proof}
 Case $s_j=1$:\\
 Let   the logarithmical work in step $j$ be denoted by $w=w^j_{\vec{s}}$,\\
let the Gibbs rescaled probability distribution after step $j$ be $p^j=p_{\vec{s}}^j$ and the one before the step $j$: $p^{j-1}=p_{\vec{s}}^{j-1}$,\\
let the occupation probabilities be $\lambda^j=\lambda_{\vec{s}}^j$ and the sum of the relevant occupation probabilities (as in definition \ref{def:wext}): $\eta^j=\eta_{\vec{s}}^j$   \\
 Let $x \in (0,aw]\bigcap(0,Z_j]$ and $b \in (0,\infty)$ such that $\int_0^b \exp\left(-\frac{E^{j-1}_{\left\lceil \frac{yn}{w} \right\rceil}}{kT}\right) \md y=x$.
 \begin{eqnarray*}  
  p^j(x)&=& p^j\left(\int\limits_0^b \exp\left(-\frac{E^{j-1}_{\left\lceil \frac{yn}{w} \right\rceil}}{ kT}\right) \md y\right)\\
  &=& p^j\left(\int\limits_0^b \exp\left(-\left(\frac{E^{j-1}_{\left\lceil \frac{yn}{w} \right\rceil}-kT\ln(w)}{kT}\right)\right) \frac{1}{w}\md y\right)\\
  &=& p^j\left(\int\limits_0^{b/w} \exp\left(-\left(\underbrace{\frac{E^{j-1}_{\left\lceil zn \right\rceil}-kT\ln(w)}{kT}}_{\stackrel{(\ast)}{=}E^{j}_{\left\lceil zn \right\rceil}/(kT)}\right)\right)\md z\right)\\
  &\stackrel{(\ast\ast)}{=}& \frac{\lambda^j_{\left\lceil \frac{bn}{w} \right\rceil}}{\exp\left(-\left(\frac{E^{j-1}_{\left\lceil \frac{bn}{w} \right\rceil}-kT\ln(w)}{kT}\right)\right)}\\
  &=& \frac{\lambda^{j-1}_{\left\lceil \frac{bn}{w} \right\rceil}}{\exp\left(-\left(\frac{E^{j-1}_{\left\lceil \frac{bn}{w} \right\rceil}}{kT}\right)\right) w \eta^j}\\
  &=& \frac{1}{w \eta^j}p^{j-1}\left(\int\limits_0^{b/w} \exp\left(-\left(\frac{E^{j-1}_{\left\lceil zn \right\rceil}}{kT}\right)\right)\md z\right)\\
  &=& \frac{1}{w \eta^j}p^{j-1}\left(\int\limits_0^{b} \frac{\exp\left(-\left(\frac{E^{j-1}_{\left\lceil \frac{yn}{w} \right\rceil}}{kT}\right)\right)}{w}\md y\right)\\
  &=& \frac{1}{w \eta^j}p^{j-1}\left(\frac{x}{w}\right)
 \end{eqnarray*}
  Where the equation $(\ast)$ follows by definition \ref{def:wext} and the equation $(\ast\ast)$ follows by definition \ref{def:gibbsr}.   \\
 One easily sees that $p^j(x) = 0$ for   $x\geq a w$, since then $\Theta_{\{0,\ldots,l\}}(x)=0$ in definition \ref{def:wext}  .\\
 The proof for the case $s_j=0$ is analogous.
\end{proof}

This next lemma shows how the partition function changes during a work extraction, as a function of how much the chosen levels are stretched (encoded in $w$) and how many levels are shifted (encoded in $a$ as described above). 

\begin{lem}
 The partition function $Z_j$ immediately after step $j$ is given by:
 \[Z_j= Z_{j-1}+a (w_1 -1),\]
 where $(0,a]$ is the interval on which the Gibbs-rescaled distribution is associated with the stretched levels, and $w_1$ is the logarithmic work extracted if the extraction is successful.
\end{lem}
\begin{proof}
Let the $(0,a]$ interval be associated with blocks corresponding to the levels $\{1,\ldots,l\}$ and  split the interval $(a,Z_j]$ into $n-l$ blocks for some $n$.
 \begin{eqnarray*}
  Z_j&=&\sum\limits_k e^{\frac{-E_k^j}{kT}}\\
    &=&\sum\limits_{k=1}^l e^{\frac{-E_k^j}{kT}} + \sum\limits_{k=l+1}^n e^{\frac{-E_k^j}{kT}}\\
    &=&\sum\limits_{k=1}^l e^{\frac{-\left(E_k^{j-1} -kT \ln\left(w_1\right)\right)}{kT}} + \sum\limits_{k=l+1}^n e^{\frac{-E_k^j}{kT}}\\
    &=&w_1  \underbrace{\sum\limits_{k=1}^le^{\frac{-E_k^{j-1}}{kT}}}_{a} + \sum\limits_{k=l+1}^n e^{\frac{-E_k^{j-1}}{kT}}\\
    &=& w_1 a - a + \sum\limits_{k=1}^l e^{\frac{-E_k^{j-1}}{kT}} + \sum\limits_{k=l+1}^n e^{\frac{-E_k^{j-1}}{kT}}\\
    &=& Z_{j-1}+a w_1 -a
 \end{eqnarray*}
out of which the lemma follows.
\end{proof}

\subsection{The work extraction game}
We consider scenarios where there is an external agent who wants to use thermalisations and work extractions to transform a system with an initial Hamiltonian $H_i$ and density matrix $\rho$, to a given final Hamiltonian $H_f$ and density matrix $\sigma$. In the process the agent will want to keep the energy of the work reservoir as high as possible, in a way that will be made more precise below.

\begin{defin}[The work extraction game]\label{def:game}
There are three systems and a work-extraction agent. 
One system is the working medium, another is a heat bath of temperature T, and the last is the work reservoir.

The initial energy spectrum $\{E\}$ of the working medium is arbitrary but given. The initial density matrix $\rho$ of the same is diagonal in the energy basis. The final energy spectrum $\{F\}$ and diagonal density matrix $\sigma$ are also arbitrary but given. 

The agent can combine thermalization (defined above) and work extraction (also defined above) in any sequence. This sequence, together with the specifications for each step is called the agent's {\em strategy}.

In a single-shot implementation of the strategy there will be a transfer of some energy $\nu$ to the work extraction reservoir. Before the extraction the agent must specify $W$. If $\nu\geq W$ and the final state conditioned on $\nu\geq W$ is $\sigma$, the work extraction is termed {\em successful} (or else a failure). The probability of success is called $1-\varepsilon$.   
\end{defin}

A crucial quantity we will be interested in calculating is the optimal work that the agent can be {\em guaranteed} to extract or need to insert. Before defining this quantity mathematically we recall a motivation for being interested in it: consider a scenario where some process is activated only if the the work reservoir energy goes above a certain threshold. One is then interested in whether this threshold is guaranteed to be exceeded. This is as opposed to the standard paradigm of focussing on the average energy increase in the reservoir. This is a key difference between the single-shot paradigm and average paradigm.

\begin{defin}[Guaranteed work]\label{def:guaranteedwork}
For a given strategy $S$, and a given initial state there is a probability distribution of work transferred to the reservoir, $p_S(\mathcal{W})$. We denote the work guaranteed up to a probability of failure $\varepsilon$ associated with that strategy as 
$W^\varepsilon_S$, and define it through the equation  
\begin{eqnarray*}
W^\varepsilon_S=\max y : \int\limits_0^y p_S(\mathcal{W})d\mathcal{W}\leq \varepsilon. 
\end{eqnarray*}

For an initial Hamiltonian $H_i$, density matrix $\rho$ and tolerated probability of failure $\varepsilon$, there is a set $\mathbbm{S}$ of allowed strategies which succeed with probability greater than or equal to $1-\varepsilon$. We denote the optimal work guaranteed (up to failure probability $\varepsilon$) for the given initial and final conditions by  
$W^\varepsilon(\rho,H_i\rightarrow \sigma, H_f)$ and define it as the optimal work over all the allowed strategies in the set:
\begin{eqnarray*}
W^\varepsilon(\rho,H_i\rightarrow \sigma, H_f):=\sup_{S\in\mathbbm{S}} W^\varepsilon_S(\rho,H_i\rightarrow \sigma, H_f).
\end{eqnarray*}
(Note that this quantity may be negative in the case where work is required to effect the given change in state and Hamiltonian).
\end{defin}

\subsection{Notation reminder}
To assist the reading of the proofs below we collect key notation in the following:
\begin{defin}[Notation] We shall use the following notation:\\
   $\vec{s} \in \{0,1\}^m$ :  a vector with one entry for each of $m$ work extractions (subsequently called ``steps''): 
    %$s_j = 0$: thermalization in the j'th step\\
    $s_j=1$: system is in one of the energy levels chosen for work extraction\\
    $s_j=0$: system is {\em not} in one of the states chosen for work extraction\\
    $\vec{s}$ is called a {\em path}.
$\hat{s}_j$ is the complement of $s_j$: $s_j=1 \Leftrightarrow \hat{s}_j=0$ and $s_j=0 \Leftrightarrow \hat{s}_j=1$\\
$w_{\vec{s}}^j$: logarithmical work ($kT\ln(w_{\vec{s}}^j) = W_{\vec{s}}^j$) extracted in step $j$ on path $\vec{s}$.\\
$w^j$: The logarithmical work one extracts in step $j$ if the specified level is occupied.\\
$W$: work demanded in order to call the total extraction {\em successful} (see definition \ref{def:game}).\\
$w = \exp ( W/(kT) )$: total logarithmical work demanded in order to call the total extraction {\em successful}.\\
$G$ is the set of successful paths, i.e. those yielding as much work as demanded:
 \[
  G = \left\{ \vec{s}\left|\prod\limits_{j=1}^m w_{\vec{s}}^j \geq w\right. \right\}
 \]
$\eta^j_{\vec{s}}$: probability of picking step j on the path $\vec{s}$.   I.e. as in definition \ref{def:wext}: $\eta_{\vec{s}|s_j=1}^{j} = \sum\limits_{i=1}^l \lambda_{\vec{s}}^{j-1}\left( i \right) $, if the chosen energy levels for work-extraction in step $j$ are $\{1,\ldots,l\}$ and $\lambda_{\vec{s}}^{j-1}$ as defined below.  \\
$P_S$: total probability of success:
$P_S = \sum\limits_{\vec{s} \in G} \prod \limits_j \eta_{\vec{s}}^j$.\\
$\lambda_{\vec{s}}^j$: occupation probabilities after step $j$ if the previous evolution of the system is given by the path $\vec{s}$.\\
  $p_{\vec{s}}^j = G\left(\lambda_{\vec{s}}^j\right)$:   Gibbs rescaled probability distribution after step j (before thermalizing) conditioned on the previous steps on path $\vec{s}$.\\
$p_{\vec{s},t}^j$: Gibbs rescaled probability distribution after step j (after thermalizing) conditioned on the previous steps on path $\vec{s}$.\\
{\em a Block}: For $a<b$ the interval $(a,b]$ is said to be a block corresponding to a level $k$, if $p_{\vec{s}}^j$ is constant on this interval $\forall \vec{s}$.\\
$q$: final Gibbs rescaled probability distribution, conditioned on successful work extraction:
\[
 q = \sum \limits_{\vec{s} \in G} p_{\vec{s},t}^m \frac{\prod\limits_{j}\eta^j_{\vec{s}}}{P_S}
\]
$B_j$: Bistochastic matrix one chooses after step $j$ by thermalizing the system (this has to be the same for all paths).\\
$E^j(x)$: Energy of the level labelled by $x$ after step $j$.\\
$\Theta_{U} (x)$: Step function associated with an interval $U$:
\[\Theta_{U} (x) =\left\{\begin{array}{ll}1&\text{: for } x \in U\\0&\text{: else}\end{array}\right.\]
\end{defin}

%%%%%%%%%%%%%%%%%%%%%END OF SETUP SECTION%%%%%%%%%%%%%%%%%%%

%% file: tech_p5_4d.tex
%%%%%%%%%%%%%%%%%%%%END OF SETUP SECTION%%%%%%%%%%%%%%%%%%%%%%%%%%%%%

\newpage
\section{Upper bounding $W^{\varepsilon}_{\mathcal{S}}$}\label{sec:upperbound}

We shall be interested in bounding $W^{\varepsilon}_{\mathcal{S}}$ given $\varepsilon$ and the initial and final conditions. We break the calculation into several lemmas which will later be combined to prove the main theorem. But firstly we give the argument for a special case of a more restricted set of strategies, in order to give the reader a sense of why relative mixedness enters as the bounding quantity. 

\subsection{Instructive special case}
Consider zero-risk work extraction such that all levels with non-zero occupation probability are shifted. Note firstly that after a work extraction by $W=kT \ln(w)$ the height of the Gibbs-rescaled probability distribution is given by $\lambda_i/\exp\left(-\left(\frac{(E_i-W)}{kT}\right)\right) 
= \lambda_i/\left(\exp\left(-\frac{E_i}{kT}\right)w\right)$, while the width gets stretched by a factor $w$. So the new Gibbs-rescaled probability distribution is given in terms of the old one as follows:
$p_{new}(x)=\frac{P_{old}(x/w)}{w}$ (see lemma \ref{lem:wext} for more details). 

Thermalization acts as a bistochastic matrix on the Gibbs-rescaled probability distribution and therefore (see \cite{HardyLP52}) $\int_0^{l} p(x) \md x \geq \int_0^{l} p_{thermalized} (x) \md x$, 
if both distributions are monotonically falling, which we will now assume w.l.o.g. Thus after a thermalization and a work extraction the following holds: 
\begin{eqnarray*}
\int\limits_0^l p_{new, thermalized} (x) \md x &\leq&  \int\limits_0^l p_{new} (x) \md x\\
&&= \int\limits_0^l \frac{p_{old} (x/w)}{w} \md x \\
&&=\int\limits_0^{wl} p_{old} (x) \md x.
\end{eqnarray*}
Inductively, after any number of work extractions and thermalizations and total work $kT \ln (w)$:
\[
\int\limits_0^{wl} p_{initial} (x) \md x \geq \int\limits_0^{l} p_{final}(x) \md x
\]
It follows that the maximal logarithmical work given the initial and final Gibbs-rescaled distributions is given by 
\[
\max w \text{ s.t. }\, \int\limits_0^{wl} p_{initial} (x) \md x \geq \int\limits_0^{l} p_{final}(x) \md x,
\]
or equivalently in terms of the cumulative distribution functions $\mathcal{F}$,
\[\max w \text{ s.t. }\, \mathcal{F}_{\text {p(initial)}} \left(\frac{x}{w}\right)\geq \mathcal{F}_{\text {p(final)}} (x)\,\,\,\forall x.\] 
This is precisely the relative mixedness defined in the main section. In~\cite{HorodeckiO11} they also arrive at the same result for the zero-risk case (starting from an a priori different model and using different arguments).

\subsection{General case}
We now turn to the general case.
We combine the two previous lemmas to gain another relation between the Gibbs rescaled distribution at steps $j$ and $j-1$. We shall use this later in an iterative manner to relate the very first and final Gibbs rescaled distributions. 
 
 \begin{lem}\label{lem:wstep_exact}
The Gibbs rescaled probability distributions at steps $j$ and $j-1$ respectively satisfy the relation 
 \[
  p_{\vec{s},t}^{j-1}(x) = \sum\limits_{k=0,1} w_{\vec{s}|s_j=k}^{j} \eta_{\vec{s}|s_j=k}^{j} p_{\vec{s}|s_j=k}^{j}\left(x w_{\vec{s}|s_j=k}^{j} + c_{\vec{s}|s_j=k}^{j}\right)
 \]
 with constants $c_{\vec{s}|s_j=1}^{j}=0$ and $c_{\vec{s}|s_j=0}^{j}=a w^j -a$.
\end{lem}
\begin{proof}
 Let $w_k=w^j_{\vec{s}|s_j=k}$, $p_k^j=p_{\vec{s}|s_j=k}^j$, $p^{j-1}=p_{\vec{s},t}^{j-1}$, $\eta_k=\eta_{\vec{s}|s_j=k}^j$.
 Let $c_0=a w_1 -a$ and $c_1=0$. % and $c_{\vec{s}|s_j=k}^{j} = c_j$. 
 Then:
 \begin{eqnarray*}
  &&\eta_0 w_0 p^j_0 (x w_0 + c_0) + \eta_1 w_1 p^j_1 (x w_1 + c_1)\\ 
  &&= \eta_0 p^j_0 (x + a w_1 - a) + \eta_1 w_1 p^j_1 (x w_1)\\
  &&= \Theta_{(a w_1, Z_j]} (x + a w_1 -a) p^{j-1}(x) + \Theta_{(0,a w_1]}( x w_1) p^{j-1}(x)\\
  &&= \Theta_{(a,Z_j-a w_1 +a]}(x) p^{j-1}(x) + \Theta_{(0,a]}(x)p^{j-1}(x)\\
  &&= \Theta_{(0,Z_{j-1}]}(x)p^{j-1}(x)
 \end{eqnarray*}
\end{proof}
We now use the above to make a statement about the relation between the {\em integrals} of the Gibbs rescaled distribution at steps $j$ and $j-1$. We show that the distribution before step $j$ majorizes the distribution after the step, even after the latter has been stretched by the logarithmical work done ($w$ in the case $s_j=1$, $1$ else). This can be seen as a generalisation of the inequality: $\int_0^l \frac{p_{old} (x/w)}{w} \md x \geq \int_0^l p_{new, thermalized} (x) \md x $ from the above special instructive case to the case where $s_j=0$ is also possible.
 \begin{lem}\label{lem:ind}
 Let $j\in\{1,\ldots,m\}$. Let $l \in (0,Z_j]$. Let $\vec{s}'\in \{0,1\}^{m-j-1}$. Define $\vec{s}_1=(s_1,\ldots,s_j,1,s'_{1},\ldots,s'_{m-j-1})$ and $\vec{s}_0=(s_1,\ldots,s_j,0,s'_{1},\ldots,s'_{m-j-1})$. Then:
\begin{eqnarray*}
  \lefteqn{\sum\limits_{\vec{s}\in \{0,1\}^j}\int\limits_0^l\tau^j_t\circ p^j_{\vec{s}_0,t} (x) \md x }\\
  &\geq \sum\limits_{\vec{s}\in \{0,1\}^j}\int\limits_0^l&\left(w^{j+1} \eta^{j+1}_{\vec{s}_1}\tau^{j+1}_t\circ p^{j+1}_{\vec{s}_1,t} (x w^{j+1})\right. \md x\\
    &&\left.+ \eta^{j+1}_{\vec{s}_0}\tau^{j+1}_t\circ p^{j+1}_{\vec{s}_0,t} (x)\right) \md x
\end{eqnarray*}
where $\tau^j_t$ is the permutation of any blocks, which maximizes the left hand side, while $\tau^{j+1}_t$ is the one which maximizes the right hand side.
\end{lem}
\begin{proof}
 Let $p_1=p^{j+1}_{\vec{s}_1,t}$, $p_0=p^{j+1}_{\vec{s}_0,t}$, $\eta_1=\eta^{j+1}_{\vec{s}_1}$, $\eta_0=\eta^{j+1}_{\vec{s}_0}$, $w=w^{j+1}$.
\begin{eqnarray*}
 \lefteqn{\sum\limits_{\vec{s}\in \{0,1\}^j}\int\limits_0^l\tau^j_t\circ p^j_{\vec{s}_0,t} (x) \md x }\\
  &=& \sum\limits_{\vec{s}\in \{0,1\}^j}\int\limits_0^l\left(\eta_1 w \tau^j_t \circ p_1(x w) + \eta_0 \tau^j_t \circ p_0(x+aw-a)\right) \md x\\
  &=& \sum\limits_{\vec{s}\in \{0,1\}^j} \int\limits_0^{l_1}\eta_1 w \tilde{\tau} \circ p_1(x w) \\
    &&+ \sum\limits_{\vec{s}\in \{0,1\}^j} \int\limits_{a}^{l +a -l_1} \eta_0 \tilde{\tau} \circ p_0(x+aw-a) \md x\\
\end{eqnarray*}
Where the first equality is exactly lemma \ref{lem:wstep_exact}.
In the second equality $l_1 \in (0, \min(a,l)]$ is a value which maximizes the right hand side of the last line and $\tilde{\tau}$ reorders $\sum_{{\vec{s}}\in \{0,1\}^j} p_1$ in descending order in $(0,aw]$ and $\sum_{{\vec{s}}\in \{0,1\}^j} p_0$ in $(aw,Z_j]$. 
This is possible since $p_1$ and $p_0$ have disjoint support, also for different $\vec{s}$, since $a$ in definition \ref{def:wext} has to be chosen independently of the path. (See lemma \ref{lem:wext}). This reordering maximizes the last line, thus it is equal to the line above.\\
After changing variables in the second integral we can translate its bounds by $-aw+l_1$, if we translate the integrand in the opposite direction applying a second permutation. Thus:
\begin{eqnarray*}
 \lefteqn{\sum\limits_{\vec{s}\in \{0,1\}^j}\int\limits_0^l\tau^j_t\circ p^j_{\vec{s}_0,t} (x) \md x }\\
  &=& \sum\limits_{\vec{s}\in \{0,1\}^j} \int\limits_0^{l_1}\eta_1 w \tilde{\tau} \circ p_1(x w) 
      \\
			&&+ \sum\limits_{\vec{s}\in \{0,1\}^j} \int\limits_{aw}^{l +a w -l_1} \eta_0 \tilde{\tau} \circ p_0(x) \md x\\
  &=& \sum\limits_{\vec{s}\in \{0,1\}^j}\int\limits_0^{l}\left(\eta_1 w \tau^{j+1} \circ p_1(x w) 
    +  \eta_0 \tau^{j+1} \circ p_0(x) \right)\md x\\
\end{eqnarray*}
Applying any bistochastic matrix $\tilde{B}$ on the probabilities $p_0$ and $p_1$ and reordering in descending order with $\tau^{j+1}_t$ afterwards, we get (we write $\tilde{B}=B\circ(\tau^{j+1})^{-1}$ for convenience, then B is again bistochastic):
\begin{eqnarray*}
  \lefteqn{\sum\limits_{\vec{s}\in \{0,1\}^j}\int\limits_0^l\tau^j_t\circ p^j_{\vec{s}_0,t} (x) \md x }\\
 &\geq& \int\limits_0^{l}\tau^{j+1}_t \circ B \circ (\tau^{j+1})^{-1}\circ \\
    &&\sum\limits_{\vec{s}\in \{0,1\}^j}\left(\eta_1 w \tau^{j+1} \circ p_1(x w) 
    +  \eta_0 \tau^{j+1} \circ p_0(x) \right) \md x\\
  &=& \sum\limits_{\vec{s}\in \{0,1\}^j}\int\limits_0^l\left(w \eta_1\tau^{j+1}_t\circ B\circ p_1 (x w)\right.\\
   && \left.+ \eta_0\tau^{j+1}_t\circ B \circ p_0 (x)\right) \md x\\
  &=& \sum\limits_{\vec{s}\in \{0,1\}^j}\int\limits_0^l\left(w \eta_1\tau^{j+1}_t\circ p^{j+1}_{\vec{s}_1,t} (x w)
    + \eta_0\tau^{j+1}_t\circ p^{j+1}_{\vec{s}_0,t} (x)\right) \md x
\end{eqnarray*}
Where the inequality follows out of the inequality $Bp \succ p$ for any bistochastic matrix $B$ and vector $p$, which is proved in \cite{HardyLP52}.
\end{proof}

The above lemma is the main ingredient for the first part of the main theorem and the rest of the proof is straightforward:
 \begin{thmnonumb}[First part of Theorem~\ref{th:main} in main body, giving the bound]
 In the work extraction game defined above, if one is given an initial density matrix $\rho=\sum_i \lambda_i\ket{e_i}\bra{e_i}$ and final density matrix $\sigma=\sum_j \nu_j\ket{f_j}\bra{f_j}$ with $\{\ket{e_i}\}$, $\{\ket{f_j}\}$ the respective energy eigenstates and both $\rho$ and $\sigma$ having finite rank, then the work $W^\varepsilon$ one can extract with certainty except with $\varepsilon$ probability respects 
% %
 \begin{eqnarray*}
W^\varepsilon_S \leq kT\ln \left(M\left(\frac{G^{(T,H_i)}(\rho)}{1-\varepsilon}||G^{(T,H_f)}(\sigma)\right)\right). 
\end{eqnarray*}
% %
 \end{thmnonumb}
\begin{proof}
 Define $p^0_{\vec{s}'}=p$. W.l.o.g. $\vec{s}'=\{0,\ldots,0\}$ (the first probability distribution is independent of the path afterwards).
 Inductively using lemma \ref{lem:ind} one gets:
 \begin{eqnarray*}
   \lefteqn{\int\limits_0^l p(x) \md x} \\
 &=& \int\limits_0^l p^0_{\vec{s}'}(x) \md x \\
 &\geq& \sum\limits_{\vec{s}\in \{0,1\}^m} \int\limits_0^{l}
    \left(\prod\limits_{j=1}^m\eta_{\vec{s}}^j\right)
    \left(\prod\limits_{j=1}^m w_{\vec{s}}^j\right) 
    \tau_t^m \circ p_{\vec{s}}\left(x \prod\limits_{j=1}^m w_{\vec{s}}^j\right) \md x\\
&=& \sum\limits_{\vec{s}\in \{0,1\}^m} \int\limits_0^{l\left(\prod\limits_{j=1}^m w_{\vec{s}}^j\right) }
    \left(\prod\limits_{j=1}^m\eta_{\vec{s}}^j\right)
    \tau_t^m \circ p_{\vec{s}}\left(x \right) \md x\\
&\geq& \sum\limits_{\vec{s}\in G} \int\limits_0^{l\left(\prod\limits_{j=1}^m w_{\vec{s}}^j\right) }
    \left(\prod\limits_{j=1}^m\eta_{\vec{s}}^j\right)
    \tau_t^m \circ p_{\vec{s}}\left(x \right) \md x\\
&\geq& \sum\limits_{\vec{s}\in G} \int\limits_0^{l w }
    \left(\prod\limits_{j=1}^m\eta_{\vec{s}}^j\right)
    \tau_t^m \circ p_{\vec{s}}\left(x \right) \md x\\
 &=& P_S \int\limits_0^{lw} q(x)\md x
\end{eqnarray*}
where $\tau_t^m$ is the permutation which maximizes the expression of the right hand side of the first inequality ($t $ stands for ``after thermalizing'', while $m$ stands for the $m$'th time one applies lemma \ref{lem:ind}). 
Therefore (with $P_S=1-\varepsilon$):
\begin{eqnarray*}
 W^{\varepsilon}&=& kT \ln(w)\\
&\leq& kT \ln\left(\max\left\{ m\left|\int\limits_0^{l} p(x_1)dx_1\geq \int\limits_0^{lm} \left(1-\varepsilon\right)q(x_2)dx_2\,\forall l\right\}\right.\right) \\
&&=kT\ln \left(M\left(\frac{G^{(T,H_i)}(\rho)}{1-\varepsilon}||G^{(T,H_f)}(\sigma)\right)\right).
\end{eqnarray*}
This proves the first part of the main theorem.
\end{proof}

%% file: tech_achievable26_9_13.tex
\section{Upper bound $W^\varepsilon$ given by relative mixedness is achievable}\label{achievable}

This section concerns the second statement of the main theorem (theorem \ref{th:main}).
We specify a protocol that achieves the bound given in theorem \ref{th:main}, i.e. it extracts $W^{\varepsilon}$ of work with a failure probability no greater than $\varepsilon$. The protocol is within the rules of the game (defined in section \ref{sec:game}). The protocol works for the initial ($\rho$) and final ($\sigma$) states taking the form $\rho=...\otimes \ket{\xi}\bra{\xi}$ and $\sigma=... \otimes \ket{\xi}\bra{\xi}$, where $\ket{\xi}$ is one of the energy eigenstates of a system with two energy eigenstates in total. This is a small restriction. It amounts to allowing the agent an extra two-level system in a known state, working as a catalyst in the sense that it aids the process but is ultimately unchanged by it.

\subsection{Guiding example}
Before giving the general protocol it is instructive to consider an example. 
We begin with a density matrix $\phi$ with energy eigenvalues $E_i(j)$, occupation probabilities $\lambda_i(j)$ and $A_i$ defined by
 $A_i(j) = \exp\left(\frac{-E_i(j)}{kT}\right)$. These are given by:
\begin{eqnarray*}
	\lambda_i&=& \left( \frac{2}{3}, \frac{1}{3}, 0 \right)\\
	A_i &=& \left( \frac{1}{3}, \frac{1}{3}, \frac{1}{3} \right)
\end{eqnarray*}
and therefore:
\begin{equation}
	p_i(x) = 
		\left\{ 
			\begin{array}{lll}
			 	2 &,\;& x \in \left(0,\frac{1}{3}\right] \\
				1 &,\;& x \in \left(\frac{1}{3},\frac{2}{3}\right] \\
				0 &,\;& x \in \left(\frac{2}{3}, 1\right] 
			\end{array}
		\right.
\end{equation}
The final state we want to reach is defined through:
\begin{eqnarray*}
	\lambda_f&=& \left( \frac{1}{2}, \frac{1}{2}, 0 \right)\\
	A_f &=& \left( \frac{1}{6}, \frac{1}{3}, 0 \right),
\end{eqnarray*}
and therefore:
\begin{equation}
	p_f(x) = 
		\left\{ 
			\begin{array}{lll}
			 	3 &,\;& x \in \left(0,\frac{1}{6}\right] \\
				\frac{3}{2} &,\;& x \in \left(\frac{1}{6},\frac{1}{2}\right] 
			\end{array}
		\right.
\end{equation}
With a risk $\varepsilon = \frac{1}{2}$ the work for this game is limited by $W= kT \ln \left(M\left(\frac{p_i}{1-\varepsilon}||p_f\right)\right)= kT\ln\left(\frac{4}{3}\right)$.
In this example we show how this amount of work can be extracted.

We first want to raise as many energy levels as we can to infinite energy, such that if we succeed (i.e. if these levels are empty and the action therefore costs $0$ work) we start with a more known state.
Unfortunately the sum of the occupation probabilities of the lowest levels will never yield exactly $\varepsilon$, so we need to change this first.

We start by raising the empty energy level to infinite energy, such that even if one mixes it completely with any other energy level it will stay empty.
Then we lower the energy of the empty level, while constantly  mixing this level with the first one. At the same time we enhance the energy of the first level, such that in total the energy of the work reservoir is unchanged with probability $1$ (the details of this action can be found below in definition \ref{def:shift} and the following lemma). We then have:

\begin{eqnarray*}
	\lambda_1&=& \left( \frac{1}{2}, \frac{1}{3}, \frac{1}{6} \right)\\
	A_1 &=& \left( \frac{1}{4}, \frac{1}{3}, \frac{1}{3} \right)
\end{eqnarray*}
\begin{equation*}
	p_1(x) = 
		\left\{ 
			\begin{array}{lll}
			 	2 &,\;& x \in \left(0,\frac{1}{3}\right] \\
				1 &,\;& x \in \left(\frac{1}{3},\frac{2}{3}\right] 
			\end{array}
		\right.
\end{equation*}
The lowest two occupation probabilities now sum up to $\varepsilon$. We enhance the energy of these two levels by doing a work extraction changing the energy of their states by $\infty$.
With probability $1-\varepsilon=\frac{1}{2}$ we get the work $0$ and the state:
\begin{eqnarray*}
	\lambda_2&=& \left( 1, 0,0 \right)\\
	A_2 &=& \left( \frac{1}{4}, 0, 0 \right)
\end{eqnarray*}
\begin{equation*}
	p_2(x) = 
			\begin{array}{lll}
			 	4 &,\;& x \in \left(0,\frac{1}{4}\right] 
			\end{array}
\end{equation*}
which in this case is a pure state (the state would not have been pure if we had chosen $\varepsilon$ to be smaller than $\frac{1}{3}$). 
With probability $\frac{1}{2}$ we get the work $-\infty$, in which case the work extraction cannot be successful in total. So in the case where the work extraction is successful the above state is the only one we need to consider.

Now we extract the work $W= kT\ln\left(\frac{4}{3}\right)$ on all the levels. This succeeds with probability $1$. The state afterwards is given by:
\begin{eqnarray*}
	\lambda_3&=& \left( 1, 0,0 \right)\\
	A_3 &=& \left( \frac{1}{3}, 0, 0 \right)
\end{eqnarray*}
\begin{equation*}
	p_3(x) = 
			\begin{array}{lll}
			 	3 &,\;& x \in \left(0,\frac{1}{3}\right] 
			\end{array}
\end{equation*}
Again we need two levels where we only have one. Acting again as defined in definition \ref{def:shift} on the first two levels we can get:
\begin{eqnarray*}
	\lambda_4&=& \left( \frac{1}{2}, \frac{1}{2},0 \right)\\
	A_4 &=& \left( \frac{1}{6}, \frac{1}{6}, 0 \right)
\end{eqnarray*}
\begin{equation*}
	p_4(x) = 
			\begin{array}{lll}
			 	3 &,\;& x \in \left(0,\frac{1}{3}\right] 
			\end{array}
\end{equation*}
The energy of the second level is now too high and we need to lower it by $kT\ln(2)$:
\begin{eqnarray*}
	\lambda_5&=& 
		\left\{ 
			\begin{array}{lll}
			 	\left( 1,0,0 \right) &,\;& \textnormal{with probability } \frac{1}{2}\\
				\left( 0,1,0 \right) &,\;& \textnormal{with probability } \frac{1}{2}
			\end{array}
		\right.\\
	A_5 &=& \left( \frac{1}{6}, \frac{1}{3}, 0 \right)
\end{eqnarray*}
The work extracted in this step is in both cases at least $0$. So by measuring whether the energy in the work-reservoir
has been enhanced by at least $W= kT\ln\left(\frac{4}{3}\right)$, we get a ``yes'' and the wanted final state with probability $\frac{1}{2}$. 

\subsection{General case}
To make the idea clearer we start giving the general algorithm and will then give the proof of the second part of the main theorem, which builds on lemmas proved later on. We assume here that we have at least $n/2$ energy levels with $0$ occupation probability,
but make sure that in the end these levels have again $0$ occupation probability (note, that this does not change the upper bound for the work). 
We assume that the levels are ordered in descending order
of their Gibbs rescaled probability. 

\begin{defin}[Work extraction algorithm]
Let $p$ and $p_f$ be Gibbs rescaled probability distributions of two states $\rho$ and $\sigma$, with the same number of levels $n$.\\
Let $\rho$, $\sigma$ have at least $n/2$ levels with occupation probabilities $\lambda_e=0$.\\
Define $W= kT \ln(M^{\varepsilon}(p,q))$.
\begin{enumerate}
 \item Do a work extraction on the levels $k+1,\ldots,n$ by $-\infty$ (such that their width becomes $0$).\\
	If there is no $k$ for which $1-\varepsilon = \sum_{i=1}^{k}\lambda(i)$: \\
  Split the level $k$ for which $\sum_{i=1}^{k-1}\lambda(i)<1-\varepsilon < \sum_{i=1}^{k}\lambda(i)$ (see the corollary to lemma \ref{lem:shift}, below).\\
 \item Make a work extraction on all levels by $W$ (i.e. stretch their Gibbs rescaled probability distributions such that it just majorizes the final one).
 \item Thermalize the obtained state to get the final state (up to permutation).
 \item Permute the levels of the obtained state such, that one gets the final state.
\end{enumerate}
\end{defin}
\begin{figure}
 \centering
 \includegraphics{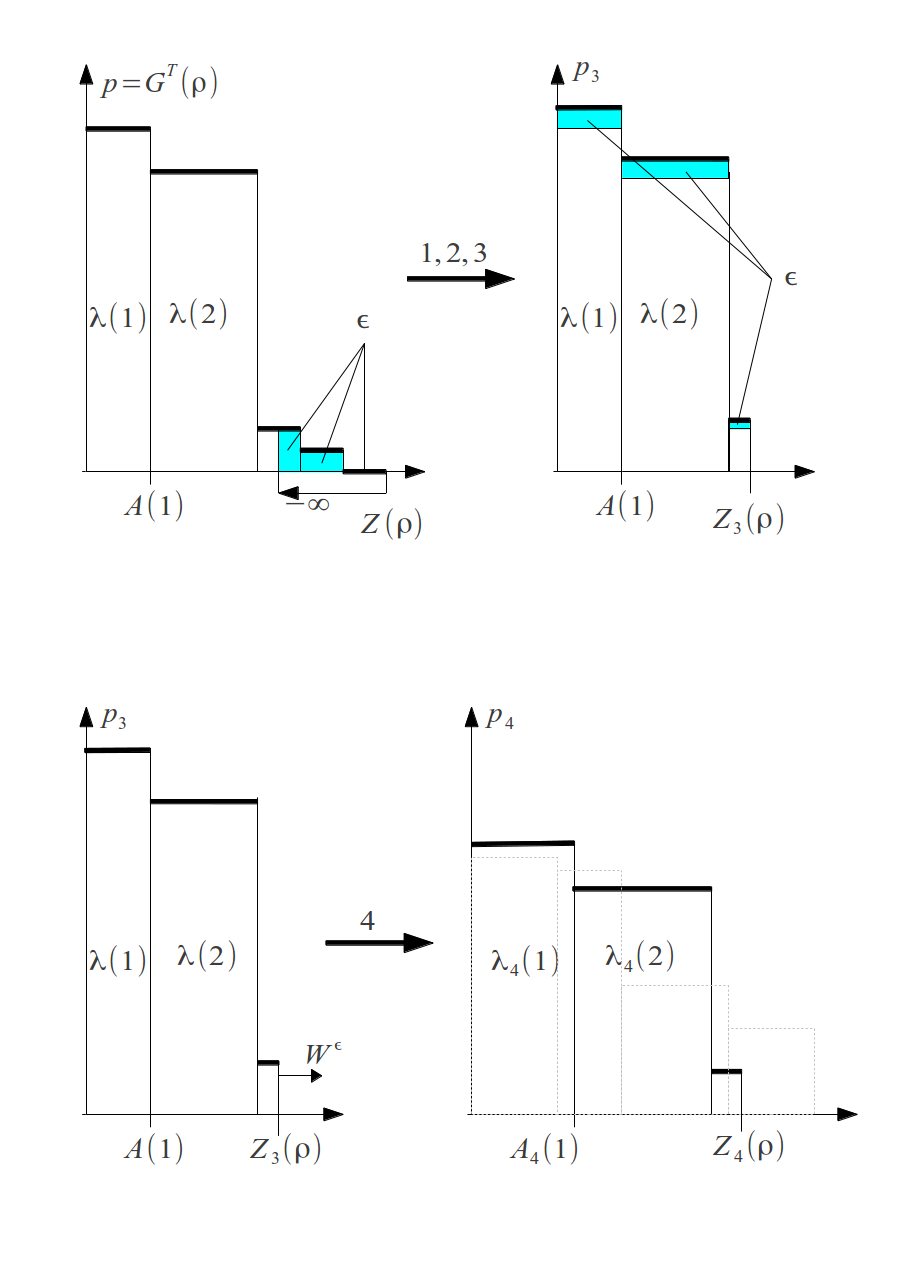}
 % work_algo.png: 900x1272 pixel, 254dpi, 9.00x12.72 cm, bb=0 0 255 361
 \caption{Work extraction algorithm: We choose the last levels such that the sum of their occupation probabilities equals $\varepsilon$, 
	then we lift them to infinity, which succeeds with probability $1-\varepsilon$ (step 1). Afterwards we extract the work $W^{\varepsilon}$ 
	and get a state which still majorises the wanted final one (step 2). 
	Thus we can get to the wanted state by doing a thermalization (step 3, see lemma \ref{lem:ass}).}
 \label{fig:walg}
\end{figure}
\begin{thm}[Bound can be achieved (second part of main theorem)]
 Let $p$ and $p_f$ be Gibbs rescaled probability distributions of two states $\rho$ and $\sigma$, with the same number of levels $n$.\\
	Let $\rho$, $\sigma$ have at least $n/2$ levels with occupation probabilities $\lambda_e=0$.\\
	Define $W= kT \ln(M^{\varepsilon}(p,q))$.\\
	The work extraction algorithm on $\rho$ yields the work $W$ with probability $1-\varepsilon$. 
	If the work extraction is successful, the final state is given by $\sigma$ with probability $1$.
\end{thm}
\begin{proof}
	The work extraction in step 1. succeeds with probability $\varepsilon$ and if it does not succeed it yields $0$ work (else $-\infty$).\\
	After step 1. the occupation probabilities are given by $\lambda_1 (i)= \frac{\lambda(i)}{1- \varepsilon}$ for $i= 1, \ldots, k$ (post-selecting on the case, in which the state was not one of the less likelier) and $\lambda_1 (i)=0$ else 
	(if the work extraction ``succeeds'' and our algorithm fails). See the corollary to lemma \ref{lem:shift}, below.\\
	After step 2. by the definition of $W$ we have that $p_2(i) \succ p_f(i)$, the extracted work is $W$. 
	Therefore one can thermalize the obtained state to get the final state $\rho$ 
	(up to permutation) with probability $1$ (see lemma \ref{lem:ass}, below).
	After the permutation (if the levels have some special physical meaning) we get the final state $\rho$ with probability 1.\\
	In total we get the final state $\rho$ with probability $1$, if the work extraction succeeds and the extracted work is $W$ with probability $1-\varepsilon$.
\end{proof}

To start, we need some algorithm which allows us to shift some probability from one level to the other, if they are in thermal equilibrium. We only want to change these two levels  (say $j$, $k$), so the sum of their occupation probabilities  remains constant 
($\lambda_j +\lambda_k=const$). Also we hope to be able to do this without needing to do any work, so we keep our total knowledge of these
levels constant. To achieve this it seems a good idea to have $p_j + p_k = const$ and constantly thermal equilibrium. This is the guiding idea for the
following algorithm. Instead of doing this (rather complicated) proof one also could have assumed that one can split levels in a physical fashion (see the corollary to the next lemma for details). Then one would have got the ``isothermal shift'' for free, by simply splitting the level $k$ in two parts and afterwards removing the level $j$. But this would have been a further assumption. So the following definition and subsequent lemma can also be seen to show it possible (in principle) to achieve a splitting of a level by just having one further empty level a heat bath and a work reservoir (which remains untouched with probability $1$).

\begin{figure}
 \centering
 \includegraphics{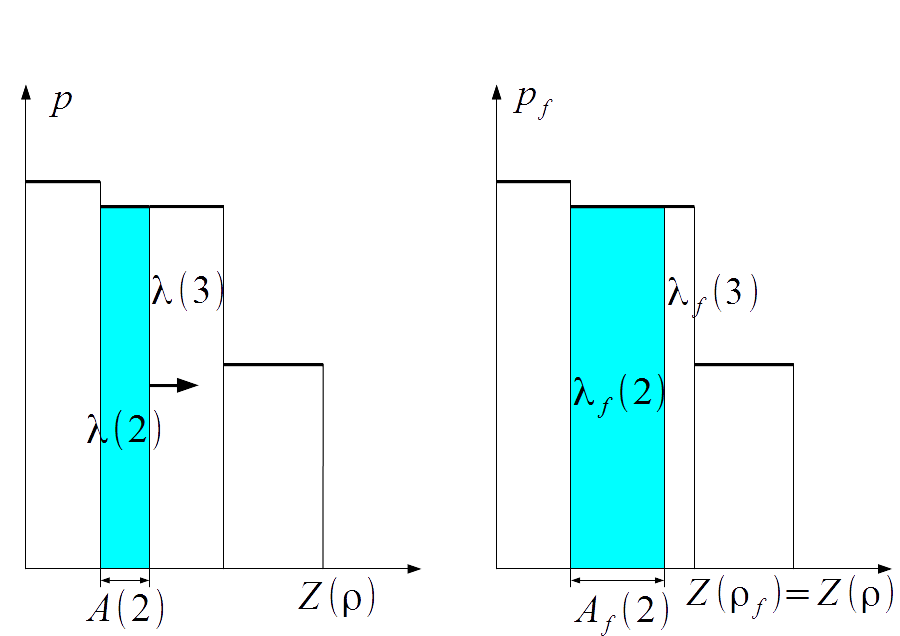}
 % shift.png: 900x636 pixel, 254dpi, 9.00x6.36 cm, bb=0 0 255 180
 \caption{Isothermal shift: The isothermal shift of the boundary between the levels 2 and 3 in direction 3 leaves
	 $p$, $\lambda_2 + \lambda_3$ and $A_2 + A_3$ invariant, while it increases $\lambda_2$ and $A_2$. The work cost is $0$.}
 \label{fig:shift}
\end{figure}

\begin{defin}[Isothermal shift of boundary]\label{def:shift}
 Let $A(j)=\exp\left(\frac{-E_j}{kT}\right)$, where $E_j$ is the energy eigenvalue of the $j$'th level.\\
 Let the levels $j$, $k= j+1$ have the same Gibbs rescaled probability.\\
 We call the limit $n\rightarrow \infty$ of the following process an isothermal shift of the boundary between $j$ and $k$ by 
$w\in\left(-\frac{A(j)}{A(j)+A(k)},\frac{A(k)}{A(j)+A(k)}
\right)$ in direction $k$:
\begin{enumerate}
 \item Do a permutation, which brings the level $j$ in front and level $k$ as second.
 \item Do a work extraction on level $j$ by:
  \[w_1=1+\frac{w}{n}\frac{A(j)+A(k)}{A(j)}\]
 \item Do a permutation, which brings the level $k$ in front and level $j$ second.
 \item Do a work extraction on level $k$ by: 
  \[w_2=1-\frac{w}{n}\frac{A(j)+A(k)}{A(k)}\]
 \item Do a thermalization totally mixing the two levels $j$ and $k$ and letting all others untouched 
	(i.e. the matrix with entries $1/2$ in $(1,1)$, $(1,2)$, $(2,1)$ and $(2,2)$ and $\delta_{m,l}$ everywhere else, 
	such that the first entry of the vector it is applied on, is the probability of the level $j$ after work extraction and 
	the second is the probability of the level $k$).
 \item Restart with 1. $n$ times in total, redefining $A(j)$ and $A(k)$ as above for the probabilities after this process.
 \item Do a permutation, which brings back the levels $j$ and $k=j+1$ at their position at the beginning (we show below, that this is possible).
\end{enumerate}
\end{defin}

Instead of the first four actions, we could have simply said we do extract the work $w_1$ on the level $j$ and the work $w_2$ on the level $k$.
Then we would have had to continue with doing the total mixing also between these levels (instead of at the first and second position of the
matrix) and so on. What we mean here with doing a work extraction on the level $j$ is the action: do a permutation bringing 
the level $j$ in front, extract work, permute the level back. 

In later definitions we will make use of 
this. Here we do not, since the algebra would get slightly more complicated.

The following Lemma shows that the above process costs no work with probability $1$ and that it can indeed be seen as a shift of the separation between the levels.
\begin{lem}[Action of the isothermal shift of boundary]\label{lem:shift}
Let $A(j)=\exp\left(\frac{-E(j)}{kT}\right)$, where $E(j)$ is the energy eigenvalue of the $j$'th level.\\
 Let the levels $j$, $k= j+1$ have the same Gibbs rescaled probability.\\
 After an isothermal shift of the boundary between $j$ and $k$ by $w\in\left(-\frac{A(j)}{A(j)+A(k)},\frac{A(k)}{A(j)+A(k)}
\right)$ in direction $k$:
\begin{enumerate}
 \item \begin{enumerate} \item the energy eigenvalues of all levels but $j$ and $k$ remain constant. 
  \item At the end $A_f(j)=\exp\left(\frac{-E_f(j)}{kT}\right)$ is given by $A_f(j)= A(j) + w (A(j)+A(k))$ and for the level $k$: $A_f(k)= A(k) - w (A(j)+A(k))$ ($E_f(j)$ is the energy of the eigenvalue $j$ after the shift).
\end{enumerate}
 \item with probability $1-(\lambda(j)+\lambda(k))$, the occupation probabilities of the final state are given by $\frac{\lambda(l)}{1-(\lambda(j)+\lambda(k))}$ for $l\neq j,k$ and $0$ for $l= j,k$. 
 \item With probability $\lambda(j)+\lambda(k)$, the occupation probabilities of the final state are given by $\frac{A_f(l)}{A(j)+A(k)}$ for $l= j,k$ and $0$ else.\\
 \item With probability $1$ the energy in the work reservoir is changed by $W=0$.
\end{enumerate}
 \end{lem}

\begin{proof}
1.(a) just follows out of the algorithm, since we did not do any work extraction on any levels and this is the only way we can change energies in
our game.
For 1.(b) we need to look at how the energy eigenvalues of the $j$'th and $k$'th level change each of the $n$ times one goes through the algorithm
in definition \ref{def:shift}. directly from the algorithm we get, that in the first time one goes through it $A(j)$ changes to $A_1(j)=\exp\left(\frac{-E(j)+kT\ln(w_1)}{kT}\right)$ and we get
$A_1(j)=w_1 A(j)= A(j) + \frac{w}{n} \left(A(j)+A(k)\right)$ and by the same argument $A_1(k)=w_2 A(k)= A(k) - \frac{w}{n} \left(A(j)+A(k)\right)$.
Since $A_1(j)+ A_1(k)=A(j)+A(k)$ we see, that after $l$ times one goes through the algorithm, one ends up with: 
$A_l(j)= A(j) + (l-1) \frac{w}{n} \left(A(j)+A(k)\right) + \frac{w}{n} \left(A(j)+A(k)\right)=A(j) + l \frac{w}{n} \left(A(j)+A(k)\right)$ and
$A_l(k)= A(k) - l \frac{w}{n} \left(A(j)+A(k)\right)$. With $l=n$ we get what is stated in 1. (b).

In order to derive 2. and 3. we need to have a closer look at how the occupation probabilities change each of the $n$ times we go through the algorithm. The occupation probabilities are given
by the Gibbs rescaled probabilities multiplied with the corresponding $A(l)$.\\
 Let $q$ be the Gibbs rescaled probability distribution after step 1. of the $i$'th time one goes through the algorithm in definition \ref{def:shift}.
 After step 2. we have:
 \[
q(x)\Rightarrow\left\{\begin{array}{ll}
                          \frac{q\left(\frac{x}{w_1}\right)\Theta_{(0, A_j]}(x)}{w_1 \eta(q_j)}&\textnormal{, with prob. } \eta(q_j)\\
			  \frac{q\left(x-A_jw_1 +A_j\right)\Theta_{(A_j,Z(q)]}(x)}{1-\eta(q_j)}&\textnormal{, with prob. } 1- \eta(q_j)\\
                         \end{array}\right.
\]
where $\eta(q_j)= \int_0^{A_j}q(x)\md x$ and $Z(q)$ is the partition function of $q$. \\
After step 4. we thus have: 
\[
    \left\{\begin{array}{ll}
		\frac{q\left(\frac{x}{w_2}\right)\Theta_{(0,A_k]}(x)}{w_2 \eta(q_k)}&\textnormal{, w. prob. } \eta(q_k)\\
		\frac{q\left(\frac{x-A_j}{w_1}+A_j-A_k w_2 +A_k\right)\Theta_{(A_k,A_k+A_j]}(x)}{w_1 \eta(q_j)}&\textnormal{, w. prob. } \eta(q_j)\\
		\frac{q\left(x-A_jw_1-A_k w_2 +A_k+A_j\right)\Theta_{(A_j+A_k,Z]}(x)}{1- \eta(q_j)-\eta(q_k)}&,\; 1- \eta(q_j)-\eta(q_k)\\
           \end{array}\right.
\]
Noting that $q(x) = q(x/w_2)$ for $x\in(0, A_k]$ and similarly for $x\in (A_j,A_k+A_j]$ and $x-A_jw_1-A_k w_2 +A_k+A_j=x$, we can rewrite this as:
\[
    \left\{
		\begin{array}{ll}
		\frac{q\left(x\right)\Theta_{(0,A_k w_2]}(x)}{w_2 \eta(q_k)}&\textnormal{, w. prob. } \eta(q_k)\\
		\frac{q(x)\Theta_{(A_k w_2,A_k+A_j]}(x)}{w_1 \eta(q_j)}&\textnormal{, w. prob. } \eta(q_j)\\
		\frac{q\left(x\right)\Theta_{(A_j+A_k,Z]}(x)}{1- \eta(q_j)-\eta(q_k)}&,\; 1- \eta(q_j)-\eta(q_k)\\
           \end{array}
		\right.
\]

Which means that after step 5. we get:
\[
 \left\{\begin{array}{ll}
		\frac{q\left(x\right)\Theta_{(0,A_k + A_j]}(x)}{w_2\eta(q_k)}\frac{w_2\eta(q_k)}{\eta(q_j)+\eta(q_k)}&,\;\eta(q_k)\\
		\frac{q(x)\Theta_{(0,A_k+A_j]}(x)}{w_1 \eta(q_j)}\frac{w_1\eta(q_j)}{\eta(q_j)+\eta(q_k)}&,\; \eta(q_j)\\
		\frac{q\left(x\right)\Theta_{(A_j+A_k,Z]}(x)}{1- \eta(q_j)-\eta(q_k)}&,\; 1- \eta(q_j)-\eta(q_k)
           \end{array}\right.
\]

For 2. note that with probability $1-(\lambda(j)+\lambda(k))$ we get after the first time one goes through the algorithm: $q_j=q_k=0$ 
(which just means, that the state is measured to be orthogonal to $j$ and $k$).
And therefore in the subsequent steps we have $\eta(q_j)=\eta(q_k)=0$. So we get with probability $1-(\lambda(j)+\lambda(k))$, the final
probability distribution: 
\[
 \frac{p\left(x\right)\Theta_{(A_j+A_k,Z]}(x)}{1-(\lambda(j)+\lambda(k))}
\]
Since the energy eigenvalues of these levels are unchanged, we get $\frac{\lambda(l)}{1-(\lambda(j)+\lambda(k))}$ 
for $l\neq j,k$ and $0$ for $l= j,k$ for the occupation probabilities, which proves 2.\\
The final Gibbs rescaled probabilities of the levels $j$ and $k$ have the same value (since we completely mix them in step 5.). Their
integral ($\int_0^{A_j+A_k}q(x) \md x$), after the first time one goes through the algorithm keeps $1$ (with probability $\lambda(j)+\lambda(k)$). 
As noticed before, $A_f(j)+A_f(k)=A(j)+A(k)$. Thus we get that with probability
$\lambda(j)+\lambda(k)$ the occupation probabilities of the levels are given by: $\frac{A_f(l)}{A(j)+A(k)}$ for $l= j,k$ and $0$ else. Which proves 3.

Suppose in the first time one goes through the algorithm the
state is orthogonal to the levels $j,k$: then the energy in the work reservoir is unchanged throughout the whole $n$ times
one goes through the algorithm and for this case, 4. follows trivially.\\ 
We now look at the other case (the case where the state is projected onto the levels
$j, k$ the first time one goes through the algorithm). \\
Let $\vec{s}\in\{1,2\}^n$. Define $\sigma(2)=1$ and $\sigma(1)=-1$. Define $\alpha_1 = \frac{A(j)}{A(j)+A(k)}$ and $\alpha_2 = 1-\alpha_1$.\\
In the $l$'th time one goes through the algorithm one either gets the logarithmical work
\begin{eqnarray*}
  w_l(1)&=&1+\frac{w}{n}\frac{A_l(j)+A_l(k)}{A_l(j)}\\
	&=&1+ \frac{w}{n}\frac{A(j)+A(k)}{A(j) + (l-1)\frac{w}{n} \left(A(j)+A(k)\right)} \\
	&=& \frac{\alpha_1+l\frac{w}{n}}{\alpha_1+(l-1)\frac{w}{n}}
\end{eqnarray*}
or the similarly derivable value for $w_l(2)$ ($A_l$ is defined in the proof of 1.(b)). Thus we can write:
\[
 w_l(s_l)= \frac{\alpha_{s_l} + \sigma(s_l) l \frac{w}{n}}{\alpha_{s_l} + \sigma(s_l) (l-1)\frac{w}{n}}
\]
In total we get the logarithmical work:
\[
 w_{tot}=\prod\limits_{l=1}^n \frac{\alpha_{s_l} + \sigma(s_l) l \frac{w}{n}}{\alpha_{s_l} + \sigma(s_l) (l-1)\frac{w}{n}}
\]
with probability (given, that we have the case where the state is projected onto the levels
$j, k$ the first time one goes through the algorithm):
\[
 P(\vec{s}|j \vee k)=\prod\limits_{l=1}^n \left(\alpha_{s_l} + \sigma(s_l) (l-1)\frac{w}{n}\right)
\]
The expectation value of $w_{tot}$ can be computed as follows (for $n<\infty$):
\begin{eqnarray*}
	E(w_{tot})&=&\sum\limits_{\vec{s}} P(\vec{s}|j \vee k) w_{tot}(\vec{s}) \\
 		&=& \sum\limits_{\vec{s}} \prod\limits_{l=1}^n \left(\alpha_{s_l} + \sigma(s_l) l \frac{w}{n}\right) \\
 		&=& \prod\limits_{l=1}^n \left(\sum\limits_{\vec{s}} \alpha_{s_l} + \sigma(s_l) l \frac{w}{n}\right) \\
		&=& \prod\limits_{l=1}^n \left(\underbrace{\alpha_1 + \alpha_2}_{=1} + \underbrace{(\sigma(l) + \sigma(2))}_{=0} l \frac{w}{n}\right)\\
		&=& 1
\end{eqnarray*}

We now look at how much the work $W= \ln(w_{tot})$ changes, if in step $l$ one replaces $s_l$ by $\hat{s}_l$ 
(remember that $s_j=0 \Leftrightarrow \hat{s}_j=1$ and vice versa):
\begin{eqnarray*}
 &&W(s_1, \ldots, s_n) - W(s_1, \ldots,\hat{s}_l,\ldots, s_n) \\
	&& = 
	\ln\left(
		\frac{\alpha_{s_l} + \sigma(s_l) l \frac{w}{n}
			}{\alpha_{s_l} + \sigma(s_l) (l-1)\frac{w}{n}
		}
	\right)
	- 
	\ln\left(
		\frac{\alpha_{\hat{s}_l} + \sigma(\hat{s}_l) l \frac{w}{n}
			}{\alpha_{\hat{s}_l} + \sigma(\hat{s}_l) (l-1)\frac{w}{n}
		}
	\right)
\end{eqnarray*}
with $c=w\sigma(s_l)$ (and therefore $w\sigma(\hat{s}_l)=-c$), $a=\alpha_{s_l}$ (and $\alpha_{\hat{s}_l}=1-a$), $x=a+c\frac{l}{n}$ and $y=1-a-c\frac{l}{n}$ we get:
\begin{eqnarray*}
 &&\left|W(s_1, \ldots, s_n) - W(s_1, \ldots,\hat{s}_l,\ldots, s_n)\right|\\
	&&= 
	\left|\ln\left(
		\frac{x\left(y+\frac{c}{n}\right)}{\left(x-\frac{c}{n}\right)y}
	\right)\right|\\
	&&=
	\left|\ln\left(
		1 + \underbrace{\frac{1}{n}\frac{cy+cx}{xy\left(1-\frac{c}{xn}\right)}}_{z}
	\right)\right|\\
	&&\leq\frac{1}{n}|z|\\ 
	&&=:q_l.
\end{eqnarray*}

Using the McDiarmid inequality~\cite{McDiarmid} we get that the probability that $W$ differs from its expectation value is bounded by:
\[P\left(|W(\vec{s})-E(W)|\geq \delta\right)
	\leq 2 \exp\left(\frac{-2\delta^2}{\sum\limits_l q_l^2}\right)
	\leq 2 \exp\left(\frac{-2\delta^2}{\frac{1}{n}|z|^2}\right)
\]
which tends to $0$ for any $\delta>0$. Therefore we get that the work in this process is given by $0$ with probability $1$, which proves 4.
\end{proof}

\begin{coro}
	Using the above lemma one can split up any level $k$ into two parts by using an empty level $e$:
	\begin{enumerate}
		\item Permuting the levels such, that the empty level $e$ comes before the level k.
		\item Doing a work extraction by $\infty $ on the level $e$ (such that its energy is $\infty$, while its width is $0$, this costs no work, since the level is empty).
		\item Do an isothermal shift of the level $e$ in direction $k$ by $w \in (0,1)$.
	\end{enumerate}
	Then by the previous lemma the final overall distribution is the same as the initial, apart from the two levels $e$ and $k$, which have now occupation probabilities:
	
	\[
		\lambda_f(e)=w \lambda(k),
		\]
		\[
		\lambda_f(k)=(1-w) \lambda(k)
	\]
	and have energies $E$ with $\exp(-E_k/kT)=A$:
	
	\[
		A_f(e)=w A(k),\]
		\[
		A_f(k)=(1-w) A(k).
	\]		
\end{coro}
The corollary directly follows from the lemma. Next we need an algorithm which makes it possible to get the end state $\sigma$ out of the initial state $\rho$, if $p \succ p_f$
(the generalization of the step $4 \rightarrow 5$ in the example).

The idea for the algorithm is that we first take the biggest eigenvalues of $\rho$, such that their area (i.e. the sum of their occupation probabilities) is equal to the biggest occupation probability ($\lambda_f(1)$ of $\sigma$). Then we mix them and make a work extraction, such that their total width (i.e. the sum of $\exp(-E(j)/kT)$) is the same as that of the final energy level $1$. then we continue with the second and so forth.

To write down the algorithm, we first need two definitions simplifying the notation:
\begin{defin}[Generalized sum]
 If $c \in \mathbb{R}$, $c\geq1$, we define $\sum_{i=1}^c d_i := \sum_{i=1}^{\lfloor c \rfloor} d_i + (c-\lfloor c \rfloor)d_{\lceil c \rceil}$.
If $c \in \mathbb{R}$, $0\leq c<1$, we define $\sum_{i=1}^c d_i :=c \cdot d_1$. 
\end{defin}
(Note that the above definition reduces to the usual sum if $c \in \mathbb{N}$).
\begin{defin}[Gibbs-equivalent and Gibbs-expanding]\label{def:geq}
	We say two tuples of $(\rho,H_i)$, $(\sigma,H_f)$ are Gibbs-equivalent (for a given temperature) if they give rise to the same Gibbs-rescaled distribution (where both are defined, $0$ else). A transform is similarly said to be Gibbs-equivalent if it changes a tuple to a Gibbs-equivalent one. Finally a transform is said to be Gibbs-expanding if it changes a tuple $(\rho,H_i)$ to another one $(\sigma,H_f)$ with $G^T(\rho) \succ G^T(\sigma)$.
\end{defin}

\begin{lem}[Optimal Gibbs-expanding transforms] 	\label{lem:ass}
	Let $\rho$, $\sigma$ be two states, diagonal in their energy-basis of dimension $n$.Let $\rho$ and $\sigma$ have at least $n/2$ empty levels. Let $G^T(\rho) \succ G^T(\sigma)$.

	Then one can transform $\rho$ into $\sigma$ with $0$ work with probability $1$.
	
	In other words: Optimal Gibbs-expanding transforms exist and yield at least $0$ work. 
\end{lem}

\begin{proof}
	W.l.o.g. let the levels of $\rho$ and $\sigma$ be ordered in descending order.\\
	Let $\lambda_{i (f)}(j)$ denote the $j$'th level of the initial (final) state.\\
	Define $a_1 \in \mathbb{R}$ as the number of needed levels of $\rho$ s.t. the total area is equal to the area at the end:
	\[
		\sum\limits_{j=1}^{a_1} \lambda_i (j) =\lambda_f (1)
	\]
	(if $a_1 \notin \mathbb{N}$ one needs to split the level $\lceil a_1 \rceil$ as in the above corollary).\\
	Define $c$ as the width of the final first level:
	\[
		\sum\limits_{j=1}^{c} A_i (j) =A_f (1)
	\]
	where $A_{i(f)}(j) = \exp(-E_{i(f)}(j)/kT)$.\\
	Now we get because of $G^T(\rho) \succ G^T(\sigma)$:
	\[
		\int\limits_0^{A_f(1)} G^T(\rho) \md x \geq \int\limits_0^{A_f(1)} G^T(\sigma) \md x
	\]
	which by $A_f (1)=\sum_{j=1}^{c} A_i (j)$ can be stated as:
	\[
		\sum_{j=1}^{c} A_i (j) \cdot \left( \frac{\lambda_i(j)}{A_i (j)} \right) \geq A_f(1) \cdot \left( \frac{\lambda_f(1)}{A_f(1)} \right) =\lambda_f(1)=\sum\limits_{j=1}^{a_1} \lambda_i (j)
	\]
	therefore: $c\geq a_1$ and finally:
	\[
		\frac{\sum\limits_{j=1}^{a_1} \lambda_i (j)}{\sum_{j=1}^{a_1} A_i (j)}
		\geq \frac{\sum\limits_{j=1}^{a_1} \lambda_i (j)}{\sum_{j=1}^{c} A_i (j)}
		= \frac{\lambda_f (1)}{A_f(1)}
	\]
	which means that one can change the energy of the first $a_1$ such that it is equal to the energy of the level $1$ at the end, with 0 risk at no cost, since either successful or not, the energy gained will be at least $0$. The occupation probabilities $\lambda$ will obviously not be changed by this (apart the total mixing of the first $a_1$ levels).\
Now we could go on and prove the same for the second level and so forth, but there is an easier way:\\
The only ingredient we needed for the above reasoning to work was $G^T(\rho) \succ G^T(\sigma)$. But this is equivalent to $G^T(\rho)-K \succ G^T(\sigma)-K$ for any constant $K$, especially for $K=\lambda_f(1)$. Explicitly:
\[
	\int\limits_0^{l} G^T(\rho) \md x -\lambda_f(1) \geq \int\limits_0^{l} G^T(\sigma) \md x -\lambda_f(1) \;\forall l
\]
Remembering $\lambda_f (1)=\sum_{j=1}^{a_1} \lambda_i (j) = \sum_{j=1}^{a_1} A_i (j) \cdot \left( \frac{\lambda_i(j)}{A_i (j)}\right)$ the above can be rewritten as:
\[
	\int\limits_{\sum_{j=1}^{a_1}A_i (j)}^{l} G^T(\rho) \md x  \geq \int\limits_{A_f(1)}^{l} G^T(\sigma) \md x  \;\forall l
\]
i.e. we get the same requirement for the remaining levels. Which means, that we can inductively apply our argument. Since the number of non-empty levels of $\sigma$ is at most $n/2$ it follows that we need at most $n/2$ empty levels to be able to split all the levels at the right place.
\end{proof}

With this lemma we can now classify the operations which cost $0$ work (with risk $0$) and their reverse also costs $0$ work: these are exactly those which do not change the Gibbs-rescaled probability distribution and are optimal:\\
From the above lemma it follows that any optimal Gibbs-equivalent transform costs no work. Secondly, if the initial and the final state are Gibbs-equivalent such a transform exists (again by the above lemma), so it is reversible. On the other hand if a transform is not Gibbs-equivalent either it or its reverse cost more than $0$ work (by the first part of theorem \ref{th:main}). 

As an aside: this, together with the triangle inequality, proves that the symmetrised version of the mixing distance $D(a,b)={{\bf\sf M}}(a||b)+{{\bf\sf M}}(b||a)\geq 0$ is a metric on the set of probability distributions on the positive reals ordered in descending order.

\begin{figure}
 \centering
 \includegraphics{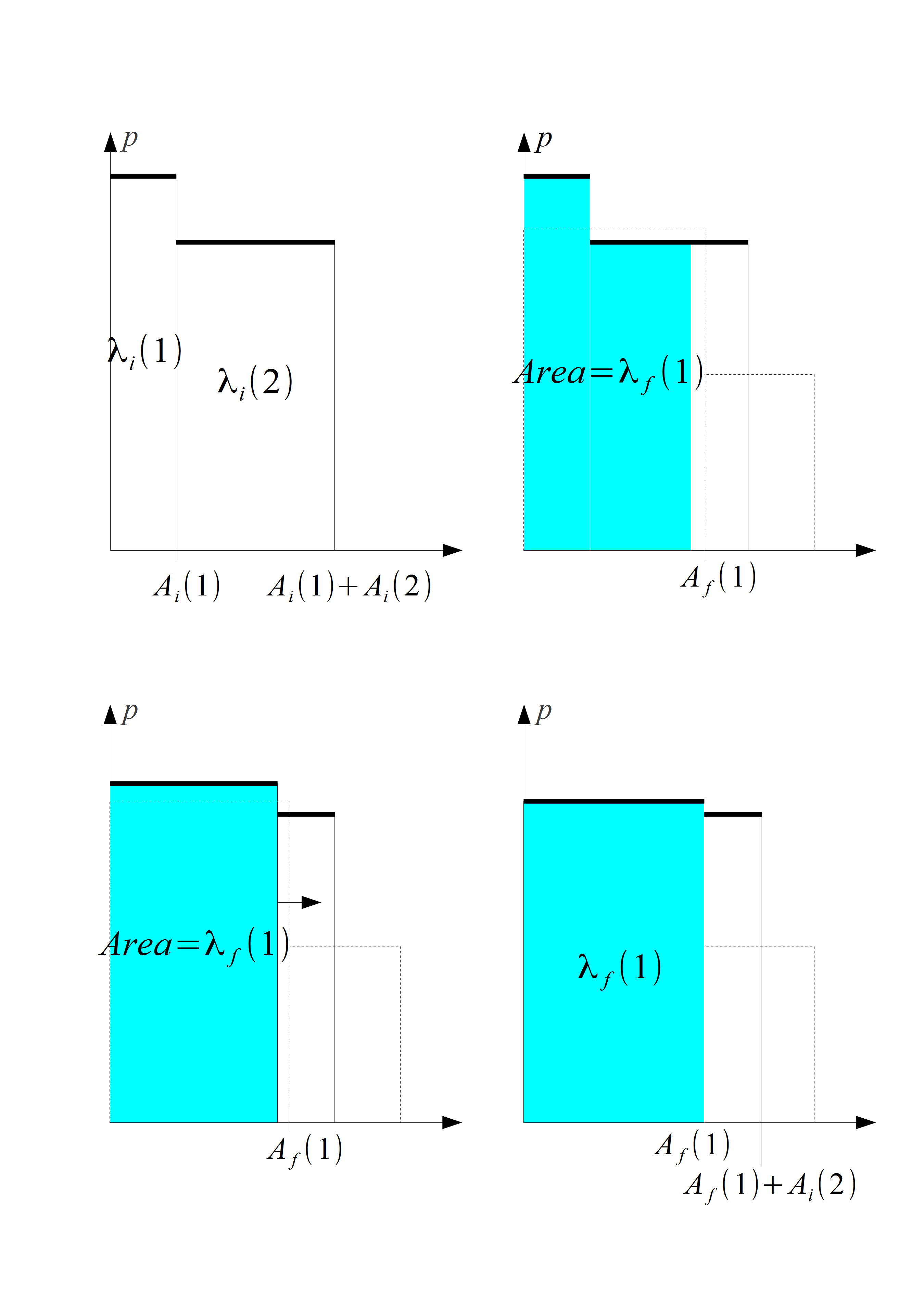}
 %% assimilation.png: 1000x1413 pixel, 254dpi, 10.00x14.13 cm, bb=0 0 283 401
 \caption{
		Gibbs-expanding transforms: One can get a state $\sigma$ out of a state $\rho$ if $p_i \succ p_f$ (with $p_i$ the Gibbs rescaled probability distribution of $\rho$ and $p_f$ that of $\sigma$), by doing the following steps for each final energy level ($j$): take as many levels (or part of levels) as needed, such that the sum of their occupation probabilities equals the occupation probability of the level $j$ (first and second pictures). Then thermalize and do a work extraction to stretch the distribution to the wanted size (third--to--fourth picture). The final $A_f(j) = \exp(-E(j)/kT)$ is bigger than the initial sum, because of $p_i \succ p_f$---therefore it is really a stretching and not a squeezing: the extracted work is at least $0$.	
	}

 \label{fig:ass}
\end{figure}

%% file: tech_2ndlaw.tex
\section{Entropy increase law}
Consider the interaction of the working medium system with the heat bath. Let $S$ be the Von Neumann entropy of the system, $\beta$ the inverse temperature associated with the bath, and $\langle E \rangle =\sum_i \lambda_i E_i$ the expected internal energy of the system. This section compares the standard law for entropy increase:
\begin{equation}
\label{eq:entropyincreaselemma}
\Delta S\geq \beta \Delta \langle E \rangle,
\end{equation}
with the one we propose should replace it: 
\begin{equation}
\label{eq:majcond}
W^0(\rho\rightarrow \rho' )\geq 0.
\end{equation}

\subsection {Our model respects standard expression} 
\begin{lem}\label{lem:entropyincrease}In the model for thermalisation used here Eq.~\ref{eq:entropyincreaselemma} is always respected.
\end{lem}

\begin{proof}
We firstly recall the model and define certain notation.

Recall that the thermalisation model states that when two levels, 1 and 2, are coupled to the heat bath, their ratio $\lambda_1/\lambda_2$ gets closer to $\exp(-\beta(E_1-E_2))$, and the other $\lambda$'s are untouched. In our model one may concatenate several such interactions to implement any allowed multi-level interaction with the bath. It will therefore suffice to show that Eq.~\ref{eq:entropyincreaselemma} holds for a single two-level interaction with the heat bath. 

For notational convenience let the probability of being in level 1 or 2 be called $\lambda_{12}:=\lambda_1 + \lambda_2$. This is then constant for the given two-level interaction with the bath. In the extreme case of the two levels interacting with the bath for an arbitrary amount of time we have $\lambda_1:=\lambda_1^{T}$ and $\lambda_2:=\lambda_2^T$ ($T$ reminds us of the temperature dependence). These values must then obey the relation 
\begin{equation}
\lambda_1^T/\lambda_2^T=\exp(-\beta(E_1-E_2))
\end{equation}
We also assume without loss of generality that $E_2\leq E_1$. This implies that $\lambda_1^T\leq 0.5\lambda_{12}$.

Now we begin to prove the statement.
Firstly we simplify $\Delta S$ by noting that only two levels change their probabilities. We write   
\begin{eqnarray*}
S&=&-\sum_i \lambda_i \log \lambda_i\\
 &=& -\lambda_1 \log \lambda_1-(\lambda_{12}-\lambda_{1}) \log (\lambda_{12}-\lambda_{1})  -\sum_{i=3}^{i_{\max}} \lambda_i \log \lambda_i\\
 &\equiv & S_{12}-\sum_{i=3}^{i_{\max}} \lambda_i \log \lambda_i. 
\end{eqnarray*}

We see that in any two-level interaction 
\begin{equation}
\label{eq:s12}
\Delta S=\Delta S_{12}.
\end{equation}

It is helpful to re-express $S_{12}$ in terms of an actual entropy $\overline{S_{12}}$, so that we can use known properties 
of entropies to make statements about $S_{12}$. We let $\overline {\lambda_1}:=\lambda_1/\lambda_{12}$ and $\overline {\lambda_2}:=\lambda_2/\lambda_{12}$ such that $\overline {\lambda_1}+\overline {\lambda_2}=1$. We define   
\begin{eqnarray*}\overline{S_{12}}:=-\overline{\lambda_1} \log \overline{\lambda_1}-\overline{\lambda_2} \log \overline{\lambda_2} .\end{eqnarray*}
One can then see in a few lines of algebra that 
\begin{eqnarray*}S_{12}=\lambda_{12}\overline{S_{12}}-\lambda_{12} \log \lambda_{12}.\end{eqnarray*}
It follows that
\begin{equation}
\Delta S_{12}=\lambda_{12}\Delta \overline{S_{12}}.
\end{equation}

We accordingly now want to show that  
$\lambda_{12}\Delta \overline{S_{12}}\geq \beta \Delta \langle E \rangle.$

We can now use a well known property of the Shannon/von Neumann entropy: $\overline{S_{12}}$ is concave in $\overline{\lambda_1}=\lambda_1/\lambda_{12}$. The function is accordingly upper bounded by any tangential line, as in Figure~\ref{fig:2ndlawboundimg}. 
\begin{figure}[h]
\begin{center}
 \includegraphics[angle=0, width=9cm, height=8cm]{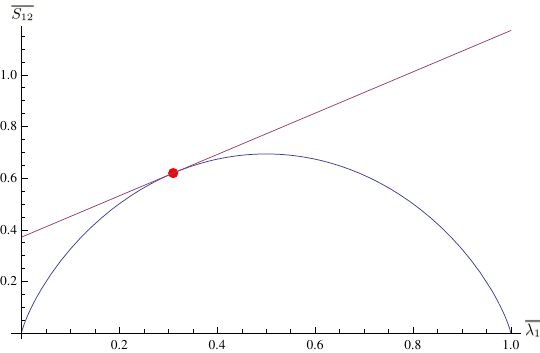} 
 \caption{The entropy $\overline{S_{12}}$ is a function of $\overline{\lambda_1}$. The red dot corresponds to the thermal state in question, i.e. $\overline{\lambda_1}=\overline{\lambda_1}^T$. The tangential upper bound has gradient $\beta (E_2-E_1)$. } 
 \end{center}
\label{fig:2ndlawboundimg} 
\end{figure}
Consider the tangential line at $\lambda_1=\lambda_1^T$. 
At that point it follows from a few lines that 
\begin{equation}
\label{eq:gradient}
\frac{d}{d\lambda_1}S_{12}|_{\lambda_1=\lambda_1^T}= \frac{d}{d\overline{\lambda_1}}\overline{S_{12}}=\beta(E_1-E_2).
\end{equation}
 
Note now that $\langle E\rangle$ may similarly to the entropy be written as  
\begin{eqnarray*}
\langle E\rangle&=&-\sum_i \lambda_i E_i\\
 &\equiv & \langle E\rangle_{12}+\langle E\rangle_{rest}, 
\end{eqnarray*}
such that $\Delta \langle E\rangle=\Delta \langle E\rangle_{12}=(\Delta \lambda_1)(E_1-E_2)$, with $\Delta \lambda_1=\lambda_1'-\lambda_1$ the change in $\lambda_1$. So $\langle E\rangle(\lambda_1)$ is a line with gradient given by 
\begin{eqnarray*}\frac{\Delta \langle E\rangle}{\Delta \lambda_1}=E_1-E_2.\end{eqnarray*}
Similarly 
\begin{eqnarray*}\frac{\Delta \langle E\rangle}{\Delta \overline{\lambda_{1}}}=\frac{1}{\lambda_{12}}(E_1-E_2).\end{eqnarray*}

Comparing this with the gradient of the tangential line to $\overline{S_{12}}$ in Eq.~\ref{eq:gradient}, we see that $\frac{1}{\lambda_{12}}\beta \langle E\rangle_{12}$ {\em has the same gradient} as the tangential line. We therefore only need to show that the change in the tangential line is upper bounded by the change in the entropy curve, as it is equivalent to  showing that $\Delta \overline{S_{12}}\geq \frac{1}{\lambda_{12}}\beta \langle E\rangle_{12}$. This must hold for all possible initial and final values of $\overline{\lambda_1}$ and all possible values of $\overline{\lambda_1}^T$(recall that we assumed without loss of generality that $\overline{\lambda_1}^T\geq 0.5$ ). These can be grouped into three cases. 

\begin{enumerate}
 \item $\overline{\lambda_1}\leq \overline{\lambda_1}^T$. Here the tangential bound above implies that $\Delta \overline{S_{12}}\geq \frac{\beta}{\lambda_{12}} \langle E\rangle_{12}\geq 0$.  
 \item ${\lambda_1}^T\leq \overline{\lambda_1}\leq 0.5$. Here the tangential bound implies that $0\geq \Delta \overline{S_{12}}\geq \frac{\beta}{\lambda_{12}} \langle E\rangle_{12}$.   
 \item $\overline{\lambda_1}\geq 0.5$, also after the interaction. Here the tangential bound implies that $\overline{\Delta S_{12}}\geq 0 \geq \frac{\beta}{\lambda_{12}} \langle E\rangle_{12}$.   
 \end{enumerate}  
This implies the lemma.
\end{proof}

\subsection{Evolutions respecting standard expression may violate Kelvin's second law}
Recall that our condition on thermalising evolutions was stronger than Eq. \ref{eq:entropyincreaselemma}. There are, as mentioned in the main body, examples of evolutions that respect Eq. \ref{eq:entropyincreaselemma} but violate our condition: Eq~\ref{eq:majcond}. In this subsection we consider whether these evolutions may violate Kelvin's second law: {\em No process is possible in which the sole result is the absorption of heat from a reservoir and its complete conversion into work.}

We use standard results concerning majorisation, as well as our main theorem. We will consider degenerate energy levels for simplicity so that Eq.~\ref{eq:entropyincreaselemma} reduces to $\Delta S\geq 0$. We now only assume that the evolution is represented by a stochastic matrix (which it is if the map is Markovian). We do not assume it is the type of thermalisation used hitherto as that would automatically respect Eq\ref{eq:majcond}. 

\begin{lem}Any stochastic matrix A which for some state violates Eq.\ref{eq:majcond} but
respects the entropy condition $\Delta S\geq 0$ will for some input state, namely the uniform distribution, violate $\Delta S\geq 0$.
\end{lem}
\begin{proof}
(i) Eq.\ref{eq:majcond} is respected iff the matrix is bistochastic. Thus A is NOT bistochastic.\\
(ii) The uniform distribution is invariant under a stochastic matrix iff it is bistochastic.
Thus A does NOT preserve the uniform distribution. Now the uniform distribution is unique in having maximal von Neumann entropy. Thus $\Delta S\geq 0$ is violated if the input state is the uniform distribution.
\end{proof}

\begin{lem} Consider a state changing to another one. Suppose: (i) the von Neumann entropy is increased, (ii) Eq.\ref{eq:majcond} is violated , and (iii) the evolution is a stochastic matrix. Then this evolution--applied to the thermal state--would allow for the violation of Kelvin's second law within our game: deterministic work extraction would be possible from a cycle where the system is in the thermal state both initially and finally.
\end{lem}
\begin{proof} Recall that we are for simplicity considering degenerate energy levels in this subsection. The thermal state is then the uniform distribution. Apply A to this (at no work cost as it represents an interaction with the heat bath). Now we have a state $\sigma$ other than the uniform distribution, so it must majorise the uniform distribution. 

To see that this implies {\em deterministic} work extraction we firstly show that $W^0 >0$ for some process using A and allowed operations within the game. Consider taking $n$ copies of $\sigma$ and going to the von Neumann limit by taking $n$ to infinity as well as taking the risk of failure $\varepsilon$ to 0. To evaluate $W^\varepsilon$ in this limit it is convenient to use Theorem~\ref{thm:max_ent} which re-expresses $W^\varepsilon$. Recall that in the von Neumann limit the smooth max entropy reduces to the von Neumann entropy $S$. We therefore have, for the case of degenerate levels:
 \begin{eqnarray*}\lim_{n\rightarrow \infty, \varepsilon \rightarrow 0} \frac{W^{\varepsilon}(\sigma^{\otimes n} \rightarrow \tau^{\otimes n})}{n}=(H_{\max}(\tau)-S(\sigma))kT\ln2,\end{eqnarray*}
where we have also used the well-known additivity of both entropies: $H_{\max}(\rho^{\otimes n})=nH_{\max}(\rho)$ and $S(\rho^{\otimes n})=nS(\rho)$ 
 In this case $\tau=\id/d$, i.e. the maximally mixed state associated with a $d$-dimensional Hilbert space. Moreover $H_{\max}(\id/d)-S(\sigma)>0$ since the uniform distribution is unique in having maximal von Neumann entropy and $H_{\max}\geq S$. Thus $W^0 >0$ for that process.

Recall secondly the subtlety that we proved that $W^{\varepsilon}(\sigma\rightarrow \sigma')$ is {\em achievable} within the game when there is access to a catalyst system. Consider extracting work from $n$ copies of $\sigma\otimes \ket{\xi}\bra{\xi}$ which will be set to $n$ copies of $\id/d\otimes \ket{\xi}\bra{\xi}$ at the end. Now $H_{\max}(\id/d \otimes \ket{\xi}\bra{\xi} )-S(\sigma \otimes \ket{\xi}\bra{\xi})>0$ as neither entropy of a state is changed by adding a pure system in this way. Thus including the catalyst system does not change the statement that $W^0>0$ for the above procedure in the von Neumann limit.  Accordingly this process violates Kelvin's law.
\end{proof}

%% file: minrelative4pages.tex
\section{Recovering the relative min-entropy}
We now show that when restricting our main theorem to the appropriate limit we recover the result of eq.~\ref{eq:AbergOppenheim} which, as discussed in the main body, was given in~\cite{Aberg11,HorodeckiO11}. Recall that this statement was
\begin{equation*}
\label{eq:AbergOppenheim2}
W^\varepsilon=kT\ln (2) D_0^\varepsilon (\sigma||\rho_T),
\end{equation*}
which should hold for the case where the final state $\rho_T$ is a thermal state on the same energy levels as the initial state $\sigma$. 

The definition of $D_0^\varepsilon(.||.)$ is as given in~\cite{Datta09} (where it is called $D_{\min}$): $D_0(\rho||\sigma):=-\log Tr(\Pi_{\rho}\sigma)$, where $\Pi_{\rho}$ is the projector onto the support of $\rho$. The smooth version is defined as $D_0^{\varepsilon}(\rho||\sigma):=\sup_{\bar{\rho}\in B^{\varepsilon}(\rho)} D_0(\bar{\rho}||\sigma)$, where $B^{\varepsilon}(\rho)$ is the set of states within $\varepsilon$ trace distance of $\rho$.   

One may first consider the special case of degenerate energy levels, as in~\cite{DahlstenRRV11} (recall that it was shown in~\cite{Aberg11} that this is a special case of~\ref{eq:AbergOppenheim}). In this case the final state (even without the Gibbs rescaling) is a uniform distribution with support $d$ at least as large as that of the initial state and taken to physically correspond to the system dimension (for $n$ qubits or bits $d=2^n$). The relative entropy expression becomes in this case 
\begin{eqnarray*}D_0^{\varepsilon}(\rho || d^{-1}\id )=\log{d}-H_{\max}^{\varepsilon}(\rho).\end{eqnarray*}
To check that this agrees with the relative mixedness expression note that the 'stretching factor' $m$ where $M(\rho||\sigma)=\log m$ is given by $m=\frac{\| \rmsupp(q)\|}{\| \rmsupp(p^{\varepsilon})\|}$. It follows that the two expressions do indeed agree in this case.

We now consider the case of non-degenerate levels. We begin with deriving the relative mixedness expression for a more general case, where the final state is some thermal state but not necessarily of the same Hamiltonian. Then we specialise to the case where it is of the same Hamiltonian, and show that the relative entropy expression is recovered.

\begin{thm}\label{thm:max_ent}
 \[
   W^{\varepsilon} = kT \ln (2) \left(H_{\max}(q) - H_{\max}^{\varepsilon}(p)\right)
 \]
where $p= G^T(\rho)$ is the Gibbs rescaled probability distribution corresponding to the initial state $\rho$ and $q=G^T(\sigma)$ is the 
one corresponding to the final thermal state $\sigma$.
\end{thm}
For the proof of this theorem a technical lemma on the smooth max-entropy is needed.
\begin{lem}\label{lem:supp_eps}
 Let $p$ be a monotonously falling probability function on $[0,\infty)$ and $d_{\varepsilon}$ be defined through 
\[
  \int\limits_0^{d_{\varepsilon}} p(x)/(1-\varepsilon) \md x =1
\]
Then:
\[
  d_{\varepsilon} = 2^{H_{\max}^{\varepsilon}(p) }
\]
\end{lem}
\begin{proof}
Let $d_{\varepsilon}$ be defined as above.
We need to show two things:
\begin{description}
 \item[i)] $\exists p^{\varepsilon}$ probability function on $[0, \infty)$ with $\| \rmsupp(p^{\varepsilon}) \| = d_{\varepsilon}$ 
   and trace-distance $\delta(p,p^{\varepsilon}) < \varepsilon$.
 \item[ii)] $\| \rmsupp(p^{\varepsilon}) \| \geq d_{\varepsilon}$ $\forall$ $p^{\varepsilon}$ monotonously decreasing probability functions on 
   $[0, \infty)$ with $\delta(p,p^{\varepsilon}) < \varepsilon$.
\end{description}
Then we get that $H_{\max}^{\varepsilon} (p) = \log_2\left(\min_{\delta(p,p^{\varepsilon})<\varepsilon} (\| \rmsupp(p^{\varepsilon}) \| )\right) = 
\log_2(d_{\varepsilon})$,
as said in the lemma.
The proof of i) goes as follows:
Define $p^{\varepsilon}(x)=p(x) \left(\int_0^{d_{\varepsilon}} p(x)\right)^{-1}$ for $x\leq d_{\varepsilon}$ and  $p^{\varepsilon}(x)=0$ for $x>d_{\varepsilon}$. 
This $p^{\varepsilon}$ is therefore normalized to one, has support $[0,  d_{\varepsilon}]$ and the following equation shows that it is also $\varepsilon$-near
to $p$:
\begin{eqnarray}
	\delta(p,p^{\varepsilon}) &=& \frac{1}{2}\left(\int\limits_{0}^{\infty} \left|p^{\varepsilon}(x) - p(x)\right|  \md x\right)\nonumber\\
	&=& \frac{1}{2}\left(\int\limits_{0}^{d_{\varepsilon}} \left|p^{\varepsilon}(x) - p(x)\right| \md x + \int\limits_{d_{\varepsilon}}^{\infty} p(x) \md x\right) \nonumber\\
	&=& \frac{1}{2}\left(\int\limits_{0}^{d_{\varepsilon}} \left(p^{\varepsilon}(x) - p(x)\right) \md x + \int\limits_{d_{\varepsilon}}^{\infty} p(x) \md x \right) \nonumber\\
	&=& \frac{1}{2}\left(1-\int\limits_{0}^{d_{\varepsilon}} p(x) \md x+ \int\limits_{d_{\varepsilon}}^{\infty} p(x)\md x \right)\nonumber\\
	&=&\int\limits_{d_{\varepsilon}}^{\infty} p(x) \md x\nonumber \\
	&<& \varepsilon \nonumber
\end{eqnarray}
which concludes the proof of i). ii) is proven on the next page (for typographical reasons).
\newpage
For the proof of ii) assume, that: $\exists p^{\varepsilon}$ like above, s.t. $\|\rmsupp(p^{\varepsilon})\| \leq d_{\varepsilon}$, then:
\begin{eqnarray*}
	\lefteqn{\frac{1}{2}\left(\int\limits_{0}^{\infty} \left|p^{\varepsilon}(x) - p(x)\right|\md x \right) }\\
	 &=&  \frac{1}{2}\left[\int\limits_{0}^{d_{\varepsilon}} \left|p^{\varepsilon}(x) - p(x)\right| \md x + 
	  \underbrace{\int\limits_{d_{\varepsilon}}^{\infty} \underbrace{\left|\overbrace{p^{\varepsilon}(x)}^{=0:\;x>d_{\varepsilon}} - p(x)\right|}_{p(x)}\md x}_{\geq \varepsilon} \right]\\
	&\geq& \frac{1}{2}\left[
		\varepsilon 
		+ 			
				\int\limits_{0}^{d_{\varepsilon}} \left(p^{\varepsilon}(x) - p(x)\right) \md x 	
		\right] \\
	&\geq& \frac{1}{2}\left[
		\varepsilon 
		+ (1-1)
		+ \underbrace{\int\limits_{d_{\varepsilon}}^{\infty} \left(p(x)-p^{\varepsilon}(x)\right) \md x}_{\geq \varepsilon}
		\right] \\
	&\geq& \varepsilon
\end{eqnarray*}
which is a contradiction to 
$\delta(p,p^{\varepsilon}) =\frac{1}{2}\left(\int\limits_{0}^{\infty} \left|p^{\varepsilon}(x) - p(x)\right| \right) < \varepsilon$.
\end{proof}
Now we have all we need to prove the theorem above:
\begin{proof}
let $p^{\varepsilon}$ be a probability function with the smallest possible support such that $\delta(p,p^{\varepsilon}) \leq \varepsilon$
and define $d_{\varepsilon}$ as in lemma \ref{lem:supp_eps}. For $l \leq d_{\varepsilon}$ the requirement for maximal work extraction reads (using the lemma)
 \begin{eqnarray*}
    \int\limits_0^l \frac{p(x)}{1-\varepsilon} \md x &\geq& \frac{l}{d_{\varepsilon}} \int\limits_0^{d_{\varepsilon}} \frac{p(x)}{1-\varepsilon} \md x
    = \frac{l}{\| \rmsupp(q)\| } \frac{\| \rmsupp(q)\| } {\| \rmsupp(p^{\varepsilon})\| } \\
&&= \int\limits_0^{l \frac{\| \rmsupp(q)\| } {\| \rmsupp(p^{\varepsilon})\| }} q(x) \md x
 \end{eqnarray*}
The above is an equation in the case $l= d_{\varepsilon}$. Which shows that the maximal $w$ as defined in theorem \ref{th:main} is given by
\[w = \frac{\| \rmsupp(q)\| } {\| \rmsupp(p^{\varepsilon})\| }= 2^{ \left(H_{\max}(q) - H_{\max}^{\varepsilon}(p)\right)}\]
\end{proof}
Eq.~\ref{eq:AbergOppenheim} is a special case of the above theorem, recovered when the final state is a Gibbs state and has also the same energy eigenvalues as the initial.
\begin{coro}
Let $\rho$ be a diagonal state with energy eigenvalues $E_i$ and $\sigma^T$ be the Gibbs state with the same energy eigenvalues $E_i$ at the bath temperature $T$. Then the maximal extractable work at risk $\varepsilon$ is given by:
  \[
   W^{\varepsilon} = kT \ln (2)  D^{\varepsilon}_{0} (\rho,\sigma^T)
 \]
\end{coro}
\begin{proof}  
Let $p$ be the Gibbs-rescaled probability function corresponding to $\rho$ and $P(j)$ the eigenvalues of $\rho$. Let 
$a$ be the flat energy probability function corresponding to $\sigma^T$. Let $A(j)=\frac{\exp\left(\frac{-E(j)}{kT}\right)}{Z}$, where $E(j)$ are the energy-eigenvalues of $\rho$ and $\sigma^T$ and $Z$ is the corresponding partition function.
This means by definition, that 
\[
p\left(Z \int\limits_0^x A
	\left(\left\lceil \frac{
	  y \cdot n}{A}
	\right\rceil
    \right) \md y \right) = 
    \frac{
      P\left(\left\lceil \frac{
	x \cdot n}{A}
      \right\rceil
      \right)
    }{
      A\left(\left\lceil \frac{
	x \cdot n}{A}
      \right\rceil
      \right)
      Z
    }
 \]
and likewise $a(x)=1/Z$ (both defined for $x\in [0,Z]$). \\ 
From the above theorem we get:
\begin{eqnarray*}
   W^{\varepsilon} &=& kT \ln (2) \left(H_{\max}(a) - H_{\max}^{\varepsilon}(p)\right)\\
&=& kT\ln(2) \left(\log_2(Z) - \log_2\left(\inf\limits_{\delta(p^{\varepsilon},p)<\varepsilon} \mathnormal{supp}\left( p \right)\right)\right)\\
&=& - kT\ln(2) \log_2 \left(\frac{1}{Z}\right. \\
&&\cdot \left.
  \min \limits_{\left\{x | \int\limits_0^x P\left(\left\lceil y \cdot n \right\rceil \right) dy > 1- \varepsilon\right\}} \left( Z  \int\limits_0^x A
	\left(\left\lceil 
	  y \cdot n
	\right\rceil \right) \md y \right)
\right)\\
&=& kT \ln (2)  D^{\varepsilon}_{0} (P,A)
 \end{eqnarray*}
\end{proof}

\newpage

\section{Triangle inequality}
The logarithmic relative mixedness respects a triangle inequality:
\begin{lem}[Triangle inequality]\label{lem:tri}
Let $\rho$, $\sigma$ be states and $\varepsilon_{1,2} \in [0,1)$\\
Let $m_1=M\left(\left.\frac{G^T(\rho)}{\varepsilon_1}\right\|G^T(\tau)\right)$ \\ and
$m_2=M\left(\left.\frac{G^T(\tau)}{\varepsilon_2}\right\|G^T(\sigma)\right)$.
\[
	M\left(\left.\frac{G^T(\rho)}{\varepsilon_1+\varepsilon_2}\right\|G^T(\sigma)\right)
	 \geq m_1 m_2 
\]
For all states $\tau$.
\end{lem}
\begin{proof}
	Let $\rho$, $\tau$ and $\sigma$ be states and $\varepsilon_{1,2} \in [0,1)$. 
	Let $m_1 =M\left(\left.\frac{G^T(\rho)}{\varepsilon_1}\right\|G^T(\tau)\right)$ and $m_2=M\left(\left.\frac{G^T(\tau)}{\varepsilon_2}\right\|G^T(\sigma)\right)$.
	Let $p=G^T(\rho)$, $q=G^T(\sigma)$ and $s=G^T(\tau)$.
	\begin{eqnarray*}
		\int\limits_0^{lm_1 m_2} q(x) \md x &\leq& \int\limits_0^{l m_1} \frac{s(x)}{1-\varepsilon_2} \md x \\
		&\leq& \int\limits_0^{l} \frac{p(x)}{(1-\varepsilon_2)(1-\varepsilon_1)} \md x \\
		&\leq& \int\limits_0^{l} \frac{p(x)}{1-\varepsilon_2-\varepsilon_1} \md x
	\end{eqnarray*}
	Therefore there is a $m \geq m_1 m_2$ such that 
	\[\int\limits_0^{l} \frac{p(x)}{1-\varepsilon_2-\varepsilon_1} \md x \geq \int\limits_0^{lm} q(x) \md x.\]
	It follows: 
	\[M\left(\left.\frac{G^T(\rho)}{\varepsilon_1+\varepsilon_2}\right\|G^T(\sigma)\right) \geq m \geq m_1 m_2.\]
\end{proof}

%% file: entanglement_measure.tex
\section{Relative mixedness as entanglement measure}\label{entanglement}
We want to start with any finite dimensional bipartite pure state $\rho_{AB}$ tensor 
  a pure entangled state of dimension $M^i$ 
and end up in any finite dimensional bipartite pure state $\sigma$ tensor   a pure entangled state of dimension $M^f$ 
  under LOCC. 
  For $M^i=2^{m_i}$ and $M^f = 2^{m_f}$, these additional states can be thought of consisting of $m_i$ ($m_f$) Bell states.  
The question is now, how many initial and final Bell states one needs to do such an operation.

Since the states are finite dimensional we can write them in the Schmidt decomposition (see e.g.~\cite{NielsenC00}):
\[
	\rho_{AB}
	= \sum\limits_{j=1}^{r^i}
		\sqrt{P_j}\ket{i_j}_A \ket{i_j}_B 
	\otimes 
	\sum\limits_{k=1}^{M^i}
		\frac{1}{\sqrt{M^i}}\ket{b_k}_A \ket{b_k}_B
\]
\[
	\sigma_{AB}
	=
	\sum\limits_{j=1}^{r^f}
		\sqrt{Q_j}\ket{f_j}_A \ket{f_j}_B 
	\otimes 
	\sum\limits_{k=1}^{M^f}
		\frac{1}{\sqrt{M^f}}\ket{b_k}_A \ket{b_k}_B
\]
By Nielsen~\cite{Nielsen99} the sufficient and necessary condition for this action being possible is:
\begin{align}
	\tilde{Q} 
	=
	\left(\underbrace{\frac{Q^1}{M^f}, \ldots, \frac{Q^1}{M^f}}_{M^f},
		\ldots,\frac{Q^{r^f}}{M^f},\ldots,\frac{Q^{r^f}}{M^f}
	\right) \nonumber\\
	\succ
	\tilde{P} 
	=
	\left(\underbrace{\frac{P^1}{M^i}, \ldots, \frac{P^1}{M^i}}_{M^i},
		\ldots,\frac{P^{r^i}}{M^i},\ldots,\frac{P^{r^i}}{M^i}
	\right)	
\end{align}
Defining:
\[
	p(x)
	= 
	\left\{
		\begin{array}{ll}
			P_j 	& ; x\in[j-1,j)\\
			0	& ; x\notin[0,r^i)
		\end{array}
	\right.
\]
such that $\int\limits_0^{l} p(x) \md x = \sum \limits_{j=1}^{l} P_j$
(and defining $q$ alike), we get that $\tilde{Q} \succ \tilde{P}$ exactly if 
\[
	\int \limits_0^{l/M^f} q(x) \md x 
	\geq \int \limits_0^{l/M^i} p(x) \md x\;
	\forall l \in \mathbb{N}
\]
i.e. the operation is possible iff $\frac{M^f}{M^i} \leq M(q||p)$.

Thus the number of Bell states needed to do such an operation is given by $\log_2(\frac{M^f}{M^i})\leq \log_2\left( M(q||p)\right)$.

It is not hard to show that the relative mixedness of entanglement is an entanglement monotone. This entanglement measure will be investigated in more detail elsewhere.